\documentclass[11 pt, onecolumn, draftcls]{IEEEtran}

\usepackage{amsbsy,color,multicol}
\usepackage{floatflt} 

\usepackage{amsmath}
\usepackage{amssymb}
\usepackage{times}
\usepackage{graphicx}
\usepackage{xspace}
\usepackage{paralist} 
\usepackage{setspace} 
\usepackage{xypic}
\xyoption{curve}
\usepackage{latexsym}
\usepackage{amsthm}
\usepackage{ifthen}
\usepackage{caption}
\usepackage{subcaption}
\usepackage{makeidx,multirow,textcomp,longtable,afterpage,setspace}
\usepackage[english]{babel}
\usepackage[latin1]{inputenc}
\usepackage[ruled,vlined]{algorithm2e}
\usepackage{booktabs}
\usepackage{comment}

\usepackage{hyperref}

\newlength\myindent
\newlength\mycolwid

{
\newtheorem{theorem}{Theorem}
\newtheorem{lemma}[theorem]{Lemma}
\newtheorem{definition}[theorem]{Definition}

\newtheorem{corollary}[theorem]{Corollary}

}

\title{\LARGE\bf Strong Attractors in Stochastic Adaptive Networks: Emergence and Characterization}

\author{Augusto Almeida Santos$^\star$, Soummya Kar$^\star$, Ramayya Krishnan$\dagger$, and Jos\'e M.~F.~Moura$^\star$

\thanks{This work was partially supported by NSF grant $\#$ CCF-$1513936$.}
\thanks{$^\star$ A.~A.~Santos, S.~Kar and J.~M.~F.~Moura are with the Dep.~of Electrical and Computer Engineering,
Carnegie Mellon University, Pittsburgh, PA 15213, USA, (augustos@andrew.cmu.edu, soummyak@andrew.cmu.edu, moura@ece.cmu.edu).}
\thanks{$\dagger$ R.~Krishnan is with Heinz College, Carnegie Mellon University, Pittsburgh, PA 15213, USA (rk2x@cmu.edu).}
}

\begin{document}

\maketitle
\thispagestyle{empty}
\pagestyle{plain}

\begin{abstract}
We propose a family of models to study the evolution of ties in a network of interacting agents by reinforcement and penalization of their connections according to certain local laws of interaction. The family of stochastic dynamical systems, on the edges of a graph, exhibits \emph{good} convergence properties, in particular, we prove a strong-stability result: a subset of binary matrices or graphs -- characterized by certain compatibility properties -- is a global almost sure attractor of the family of stochastic dynamical systems. To illustrate finer properties of the corresponding strong attractor, we present some simulation results that capture, e.g., the conspicuous phenomenon of emergence and downfall of leaders in social networks.
\end{abstract}

\section{Introduction}

We propose a family of models to study the evolution and long term formation of networks of interacting agents whose ties vary over time as a result of their interaction. The connection between two agents is assumed to be reinforced or penalized due to their interaction (or lack of it). Namely, if an agent~$i$ attempts cooperation with an agent~$j$ -- or simply,~$i$ calls~$j$ -- and~$j$ is cooperative -- or simply,~$j$ responds to~$i$, -- then the tie is reinforced, otherwise it is penalized. If the connection is not excited, i.e.,~$i$ does not call~$j$, then the tie fades away. In Section~\ref{sec:probform}, we detail our model.

To illustrate, we observe that videos on media networks such as Youtube or products at Amazon are networked by the recommended list display: whenever one browses a video or purchases a product, one is shown a list of recommended videos or products as depicted in Fig.\ref{fig:videosnet}. It is natural to expect that once a viewer selects one of the recommended videos, the link from the video to the original one should be somehow reinforced, otherwise, it is penalized. In other words, the video is more likely to show up in the recommended list in future views of that video. The same happens with products at Amazon and the like\footnote{While our focus is not on the exact inner working mechanism of Amazon, Youtube, or any specific social media network, we assume that, in broad terms, this should be the prevailing dynamical law for the evolution of ties of any unbiased (or not so biased) algorithm running underneath such platforms.}. The study of such systems shed light, e.g., on the often elusive and universally observed viral behavior. An empirical study of this viral behavior on Youtube can be found~\cite{viralyoutube}.
\begin{figure} [hbt]
\begin{center}
\includegraphics[scale= 0.6]{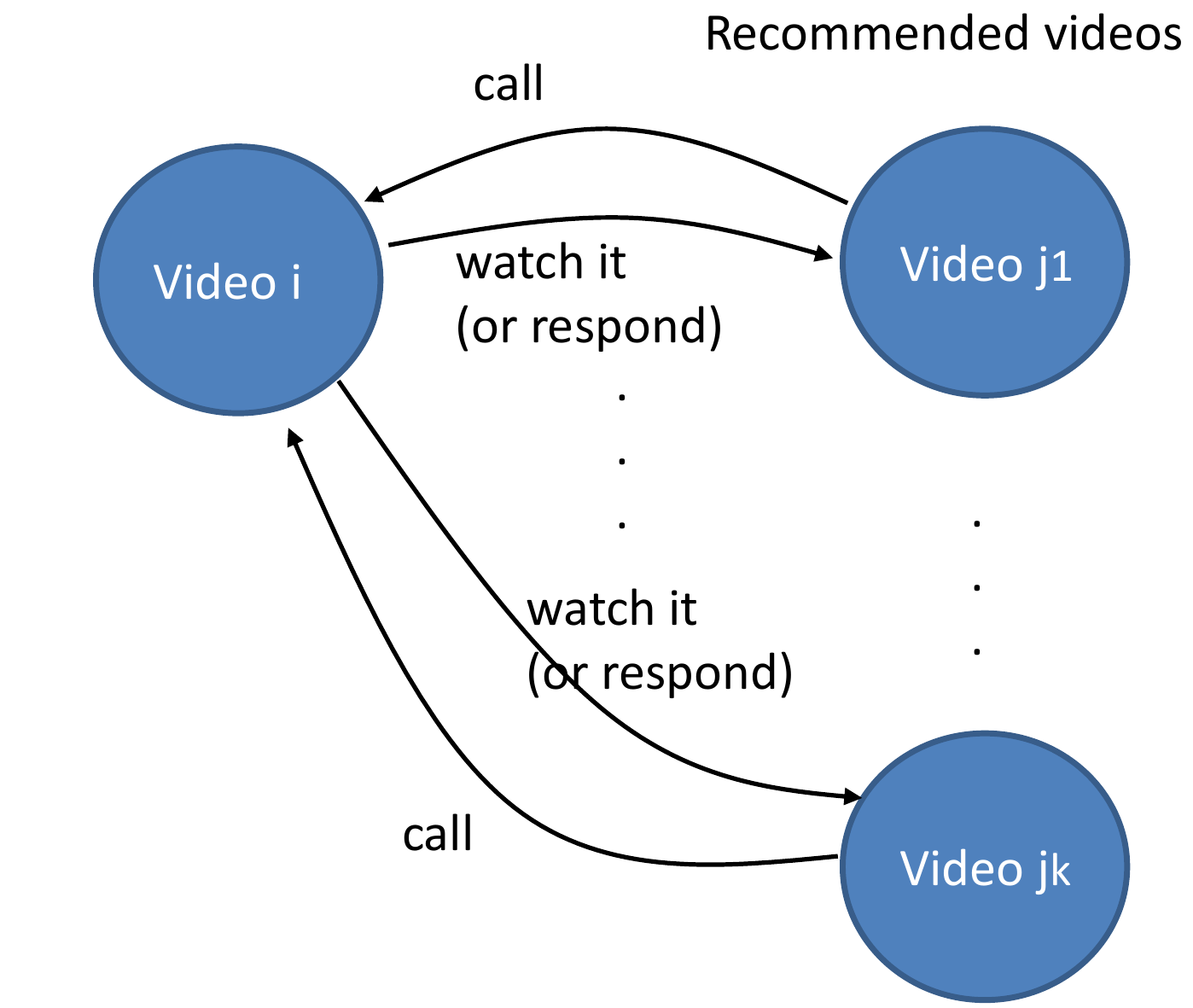}
\caption{Illustration of ties among videos on a videostreaming platform. Once a viewer browses a video~$i$, there are \emph{calls} from other videos in the recommended list. It is natural to expect that if the viewer watches at least one of the videos (i.e., responds to the call) then, the corresponding tie is reinforced.}\label{fig:videosnet}
\end{center}
\end{figure}

Other examples of networked systems whose ties change over time by reinforcement, penalization, and fading are:




$\bullet$ \textbf{Social Networks.} Individuals exert influence on other individuals for various reasons; one possible cause may be by reinforcement: if individual~$A$ calls individual~$B$, and~$B$ responds, then there is a reinforcement of the tie from~$A$ to~$B$. \emph{Call} and \emph{response} here may abstract broader considerations, e.g., individual~$A$ reads a book from author~$B$ (call) and~$A$ is satisfied with the experience (response) -- in this case, the reinforcement of the tie means that the reader is more likely to read future books of the author, or prone to be influenced by the opinions of that author. Moreover, a gossip effect may take place:~$A$ can now recommend~$B$ to other individuals that may also reinforce or undermine their connection to~$B$. Other aspects of reinforcement of ties in social networks, such as homophily, are considered in the reference~\cite{krishnan}. One important aspect to address in these systems is the observed emergence of few opinion leaders/makers \emph{versus} a crowd of non-influential followers. Reference~\cite{yingdalu} addresses the emergence of opinion leaders.



\textbf{Emergent Multi-Organizational Network (EMON).} References~\cite{ritahurricane},~\cite{anisya}~\cite{katrina} emphasize the fact that natural disasters challenge centralized management emergency systems, and a cooperative network of many different organizations is crucial to handle the critical demands during severe crisis. Such emergent multi-organizational network (EMON) evolves over time and it may lead to the emergence of leaders, i.e., organizations with high degree of connectivity, betweeness or closeness centrality (as specially leveraged in~\cite{katrina} for the Katrina Hurricane). These organizations play an important role on the overall distributed decision making of the network, and the dynamical laws for the evolution of such collaborative networks clearly emerge from a reinforcement-penalization collaborative principle.

In this paper, we develop a simple family of stochastic dynamical models on the edges of a graph combining reinforcement and penalization rules, and we show that this family exhibits strong convergence properties, in the limit as time goes to infinity, namely, the underlying evolving network of tie strengths converges almost surely to a subset of binary matrices (or unweighted directed graphs). In other words, a collaborative network emerges almost surely from the various local interactions among agents. To prove our results, we develop a stochastic LaSalle-like principle: i) show that certain binary matrices (or directed graphs) are \emph{local} strong-attractors to the dynamical system; and ii) show that the system tends to be attracted to a neighborhood of these graphs.

Time-varying networks model complex networked systems whose underlying topology evolves over time. Several models for time-varying networks in the literature are built mostly to capture robust self-organizing behavior -- universally observed in nature -- as in the adaptive coevolutionary networks~\cite{adaptivenet,adaptive2,adaptive3,adaptive4}, or to estimate the underlying network evolution~\cite{kolar,Xing}, or, in the case of infinite graphs, are built upon certain tractable regularity assumptions induced by invariant properties -- such as exchangeability and \emph{c\`adl\`ag} property -- as in~\cite{Crane,Crane2}. Our family of models may be framed in the context of adaptive networks, due to its qualitative property that the strength of ties may affect the state of nodes and vice-versa.  The field of adaptive networks, in particular, is mostly concerned with the study of emergence of robust self-organizing behavior in complex systems. We illustrate in Section~\ref{subsec:rise}, via numerical simulations, that the emergent networks of collaborative ties from our family of models exhibit an important robust self-organizing configuration: the emergence of new latent leaders and downfall of veteran ones. Indeed, emergence and downfall of leaders leads to real-time `division of labour' (to quote the term from~\cite{adaptivenet}) in EMON systems, -- in that high degree, betweeness, or centraility nodes may play major roles in overall distributed decision-making, -- new viral videos on Youtube, or popular books at Amazon. The referred phenomenon is an important factor that influences financial markets and the economy as well.

To conclude, in this paper, we develop a family of stochastic adaptive network models: i) that exhibits the property of almost sure convergence to a graph in the long run (to the best of our knowledge, this property is either absent or has not been analytically established in other models for adaptive networks); ii) that captures some of the important features of reinforcement and penalization of ties associated with the dynamics of social networks; iii) that exhibits a robust self-organizing behavior associated with the emergence and downfall of leaders; and iv) is broad enough to find applications in various scenarios (other than the ones mentioned in this section) such as Synapse plasticity -- from the Hebbian plasticity principle~\cite{hebb} --, and the emergence of network trail pheromone in ant colonies (e.g.,~\cite{debora1},~\cite{antevolution},~\cite{Dorigo}).

\textbf{Outline of the paper.} Section~\ref{sec:probform} introduces the family of models. Section~\ref{sec:DSA} presents definitions, notation, and discusses our approach to the problem, namely, it provides a sketch of the proof of the main results. Section~\ref{sec:convresults} analytically establishes that the dynamical system exhibits a strong attractor, a proper subset of the binary matrices. Section~\ref{sec:simula} provides simulation results that illustrate finer characterization of the resulting attractors and interesting long-term behavior properties of the system such as the emergence and downfall of latent leaders in social networks. Section~\ref{sec:concluding} concludes the paper.


\section{Problem Formulation}\label{sec:probform}

In this section, we propose a model for the evolution of ties among a set of interacting agents. The interactions are performed in two instances: \emph{calls} and \emph{responses}, and, as a result, the tie between two nodes is reinforced or penalized. More precisely, the dynamics evolve as follows: at each time step~$n$, if node~$i$ calls~$j$ and~$j$ does not respond, then the tie~$p_{ij}(n)\in\left[0,1\right]$ from~$i$ to~$j$ is penalized, i.e.,~$p_{ij}(n+1)<p_{ij}(n)$ (meaning that~$i$ is less likely to call~$j$ at time~$n+1$), otherwise, the tie is reinforced, i.e.,~$p_{ij}(n+1)>p_{ij}(n)$. In other words, if node~$i$ calls node~$j$ and there is (respectively, there is no) response, then~$i$ is more (respectively, less) likely to call~$j$ in the next iteration. We also assume fading due to idleness in the connection~$ij$, i.e., if~$i$ does not call~$j$, then~$p_{ij}(n+1)<p_{ij}(n)$. This Reinforce-Penalize-Fade (RPF) is the defining building block of our family of dynamical systems of interacting networked agents. Fig.~\ref{fig:callres} illustrates the dynamical model.
\begin{figure} [hbt]
\begin{center}
\includegraphics[scale= 0.6]{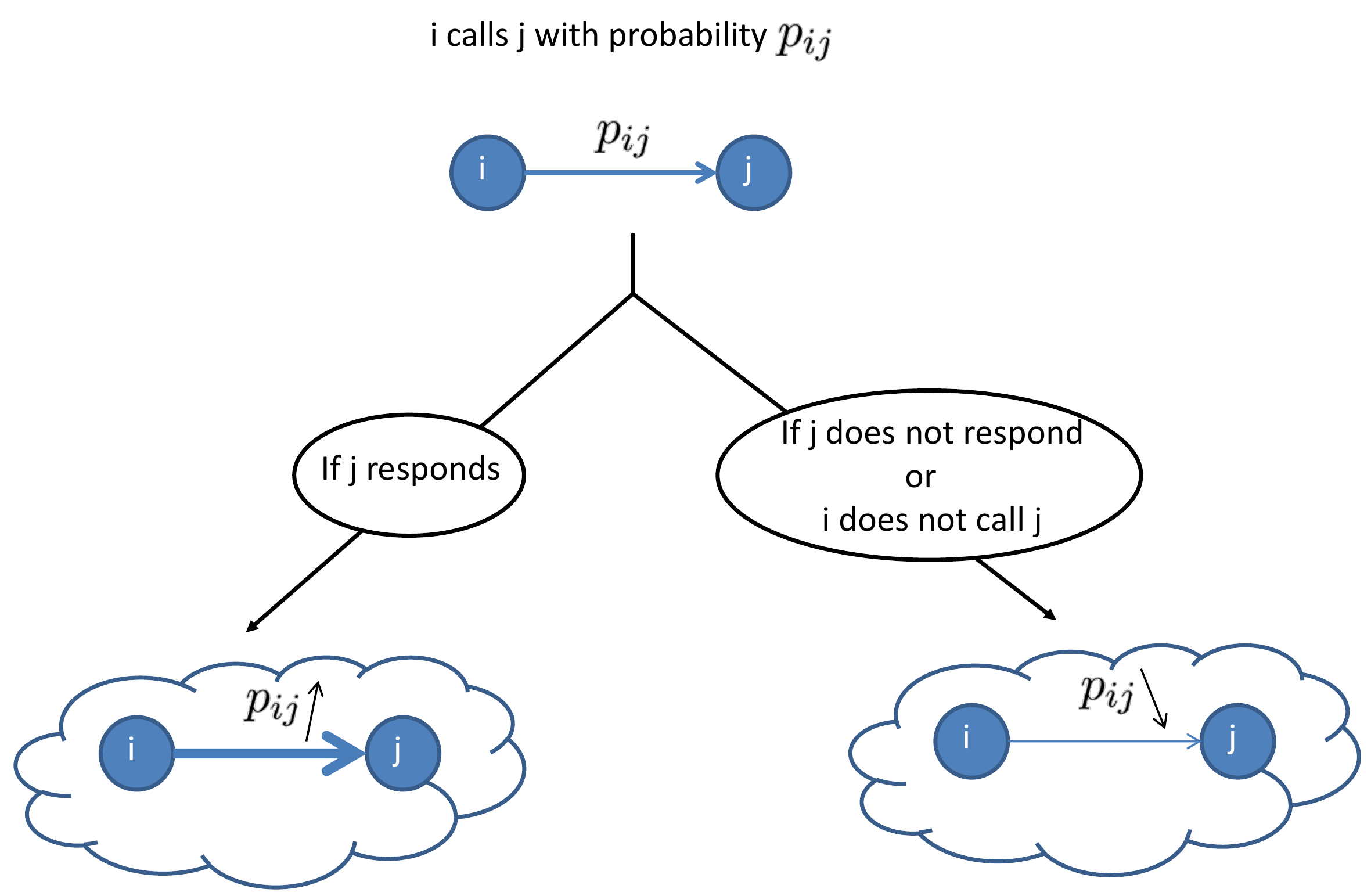}
\caption{Illustration of the reinforcement or fading of ties among agents: if node~$i$ calls~$j$ and~$j$ responds, then the tie captured by the weight~$p_{ij}$ increases, otherwise, the tie~$p_{ij}$ decreases.}\label{fig:callres}
\end{center}
\end{figure}

We model these dynamics formally as follows:
\begin{eqnarray}
p_{ij}(n+1) & = & C_{ij}(n)\left[R_{ji}(n)\underbrace{f_{i,+}(p_{ij}(n))}_{\mbox{Reinforce}}+(1-R_{ji}(n))\underbrace{f_{i,-}(p_{ij}(n))}_{\mbox{Penalize}}\right]\nonumber\\
& &+ \left(1-C_{ij}(n)\right)\underbrace{g_{i,-}\left(p_{ij}(n)\right)}_{\mbox{Fade}}\label{eq:stoc22}
\end{eqnarray}
where~$f_{i,+}$ and~$f_{i,-}$ represent the reinforcement and penalization rules at node~$i$, respectively;~$g_{i,-}$ is the fading law triggered when no call from~$i$ to~$j$ is performed; the state variable~$C_{ij}(n)$ is the indicator of a call from node~$i$ to node~$j$ at time~$n$, i.e., $C_{ij}(n)=1$ if~$i$ calls~$j$ at time~$n$, or zero, otherwise;~$R_{ji}(n)=1$, if~$j$ responds to~$i$ at time~$n$. Further comments and assumptions follow:

$\bullet$ \textbf{(Reinforcement)} ~$f_{i,+}\left(x\right)>x$ for all~$x\in\left[\left.0,1\right)\right.$ with~$f_{i,+}$ being an increasing\footnote{To be precise, the term increasing refers to strictly increasing in this paper.} function and~$f_{i,+}(1)=1$;

$\bullet$ \textbf{(Penalization or Fading)} ~$f_{i,-}\left(x\right)< x$ for all~$x\in\left(\left.0,1\right]\right.$ with~$f_{i,-}$ being an increasing function and~$f_{i,-}(0)=0$, and the same goes for~$g_{i,-}$;

$\bullet$  \textbf{(Calls)} ~$\left(\mathbf{C}(n)\right)=\left(C_{ij}(n)\right)_{ij}$ is the random process associated with calls over time and such that
\begin{equation}
\mathbb{P}\left(\left.C_{ij}(n)=1\right| \left(\mathbf{p}(k)\right)_{k\leq n}\right)=\mathbb{P}\left(\left.C_{ij}(n)=1\right| \mathbf{p}(n)\right)=p_{ij}(n),\nonumber
\end{equation}
i.e.,~$\mathbf{C}(n)$ is \emph{memoryless} and, in a sense, it is a Bernoulli random variable at each time step~$n$ conditioned on~$\mathbf{p}(n)$, and ~$p_{ij}(n)$ is in, \emph{latus sensus},  the probability that a call from~$i$ to~$j$ will be performed at time~$n$. As it will be clear in Section~\ref{sec:convresults}, we will also consider (fading) perturbations on the conditional distribution of~$\mathbf{C}(n)$;

$\bullet$ \textbf{(Responses)} We assume that nodes have limited capacity and tend to utilize all of their bandwidth for response. Specifically: i) (limited capacity) each node~$i$ can only respond to at most~$Q_{i}$ calls at each time~$n$; ii) if the number of incoming calls to~$i$ is below its capacity~$Q_{i}$, then~$i$ responds to all callers. More compactly,
\begin{equation}
\sum_{\ell} R_{i\ell}(n)=\min\left\{\sum_{\ell} C_{\ell i}(n),\, Q_{i}\right\}.\label{eq:respoconstraint}
\end{equation}
We also assume memorylessness in the response policy. Let~$R_{i}(n)$ be the~$i$-th row of~$R(n)$ and~$C_{\cdot\,i}(n)$ be the~$i$-th column of~$C(n)$. Then,
\begin{eqnarray}
\mathbb{P}\left(R_{i}(n)= v\left| \left(\mathbf{p}(k)\right)_{k\leq n}\right.\right) & = & \mathbb{P}\left(R_i(n)= v\left| \mathbf{p}(n)\right.\right) \label{eq:lasteq0}\\
& = & \sum_{u\geq v} \underbrace{\mathbb{P}\left(R_i(n)= v\left|\mathbf{p}(n), C_{\cdot\,i}(n)=u\right.\right)}_{=: M_{i}\left(v,\mathbf{p}(n),u\right)}\mathbb{P}\left(C_{\cdot\,i}(n)=u\left|\mathbf{p}(n)\right.\right)\label{eq:lasteq}
\end{eqnarray}
for any binary vector~$v\in\left\{0,1\right\}^{n}$, and the sum in~\eqref{eq:lasteq} is over all binary vectors~$u$\footnote{The inequality~$u\geq v$ is to be interpreted componentwise.} such that the support of~$v$ is contained in the support of~$u$. This way of decomposing the conditional probability expression will enable us to integrate the constraint in~\eqref{eq:respoconstraint} readily.
In other words, the conditional (on~$\mathbf{p}(n)$) distribution of~$\left(R_i(n)\right)$ is uniquely characterized by the choice of a particular function~$M_i$. To sum up, as explained below, the only requirements on the response~$\left(R_i(n)\right)$ selection process for each~$i$ are the capacity constraint in~\eqref{eq:respoconstraint} and the memorylessness property. The conditional distribution of~$\left(R_i(n)\right)$ is indexed by the class~$\mathcal{M}_{i}$ of functions
\begin{equation}
M_i\,:\,\left\{0,1\right\}^{N} \times \left[0,1\right]^{N \times N} \times \left\{0,1\right\}^{N}\rightarrow \left[0,1\right]
\end{equation}
with the following restrictions:

i) $M_i\left(v,\mathbf{p},u\right)\geq 0$ and~$\sum_{v} M_i\left(v,\mathbf{p},u\right)=1$;

ii) if~$v=u$ and~$\sum_{\ell} u_{\ell} \leq Q_i$ then,
\begin{equation}
M_i\left(v,\mathbf{p},u\right)=1;
\end{equation}

iii) if~$\sum_{\ell} v_{\ell} > Q_{i}$, then
\begin{equation}
M_i\left(v,\mathbf{p},u\right)=0.\label{eq:Mequations}
\end{equation}


\emph{Note on the response policy:} A specific response policy for a node~$i$ is fully characterized by one specific choice function in the class of functions~$\mathcal{M}_i$. The convergence results presented in this paper hold for any such choice. In particular, if a node receives more calls than its capacity, the specific mechanism or (possibly random) algorithm to select the callers to respond to are not relevant for the validity of the theorems presented.

Equation~\eqref{eq:stoc22} seems to lead to a decoupled dynamical system~$\left(\mathbf{p}(n)\right)$. In fact, the system is coupled due to the limited capacity of response~$Q_{i}$ of each node~$i$. For instance, whether a node~$i$ reinforces -- by responding to a call, -- the connection~$p_{ji}$ from a neighbor~$j$, depends on the other callers to~$i$, i.e., on~$\left\{p_{ki}\right\}_k$ as only~$Q_{i}$ connections can be reinforced. One may, at first, still argue that this conforms to only a \emph{local} coupling, i.e., the only coupling is among the edges~$\left\{p_{ki}\right\}_k$ pointing to~$i$, that is, one could study the dynamics~\eqref{eq:stoc22} by looking independently to the evolution of each column of~$\left(\mathbf{p}(n)\right)$. Even though this holds when each node~$i$ selects unbiasedly (uniformly randomly)~$Q_{i}$ callers to respond to, this is not true in general. To see this, whether~$p_{ij}$ is reinforced or not, affects the out-flow degree~$\sum_{\ell} p_{i\ell}$ of node~$i$. If, in turn, the response law is out-flow degree biased, then this will affect the evolution of the other out-flow probabilities~$p_{ik}$. But,~$p_{ik}$ is also coupled with the other weights~$p_{\ell k}$ pointing to~$k$ by the local argument just mentioned. In other words,~$p_{ij}$ is dependent on~$p_{\ell k}$ and~$\left(\mathbf{p}(n)\right)$ is, in general, a coupled dynamical system whose qualitative behavior may not be studied by partitioning the set of state variables and tracking each partite separately. In summary, under unlimited capacity, the system is uncoupled; limited capacity under an unbiased response leads to a column-decoupled system; but in general, the system is coupled.

From the memoryless assumptions on the state variables~$\left(\mathbf{C}(n)\right)$ and~$\left(\mathbf{R}(n)\right)$, and the dynamics~\eqref{eq:stoc22}, we note that~$\mathbf{p}(n)$ is a Markov process. This section constructs a family of stochastic dynamical systems on the Markov process~$\left(\mathbf{p}(n)\right)$, broadly determined by: i) the Reinforce-Penalize-Fade (RPF) rules~$f_{i,+}$,~$f_{i,-}$ and~$g_{i,-}$; and ii) the class of functions~$\mathcal{M}$ just described on the response~$\mathbf{R}(n)$. Our goals for the rest of the paper are to establish two convergence results, in Section~\ref{sec:convresults}, associated with the long-term behavior of the family of dynamical systems~\eqref{eq:stoc22}; and to explore, via numerical simulations in Section~\ref{sec:simula}, the finer aspects of the attractors of the system that capture many of the features of the evolution of social networks -- such as emergence of latent leaders and downfall of veteran ones.

In what follows, we assume that all processes and random variables are defined on a rich enough probability space~$\left(\Omega,\mathcal{F},\mathbb{P}\right)$. Whenever we refer to the set~$\Omega$, we take it to be as the probability space~$\left(\Omega,\mathcal{F},\mathbb{P}\right)$. For instance, \emph{almost all}~$\omega$ means in fact~$\mathbb{P}$-almost all~$\omega\in \Omega$. We will also refer to~$\left(\mathcal{F}_k\right)_{k\leq n}$ as the natural filtration of the process~$\left(\mathbf{p}(n)\right)$ up to time~$n$.



\section{Dynamical Systems Approach}\label{sec:DSA}

The stochastic dynamical system~\eqref{eq:stoc22} is captured compactly by the following stochastic recursive equation
\begin{equation}
\mathbf{p}(n+1)= \mathbf{f}\left(\mathbf{p}(n)\right).\label{eq:gener}
\end{equation}
We may refer to the map~$\mathbf{f}$ as a random (discrete-time) dynamical system since it defines a (discrete-time) stochastic process~$(x_n)_n$ as
\begin{eqnarray}
x_1 & = & \mathbf{f}(x_0,\omega)\nonumber\\
x_{n} & = & \mathbf{f}(x_{n-1},\omega),\nonumber
\end{eqnarray}
for each~$\omega \in \Omega$, or, in other words,~$x_n=\mathbf{f}^{n}\left(x_{0},\omega\right)$, where~$\mathbf{f}^{n}:=\underbrace{\mathbf{f}\circ \mathbf{f} \circ \cdots \circ \mathbf{f}}_{n\mbox{ times}}$ is the $n$-fold iterate of~$\mathbf{f}$ with respect to the same realization~$\omega$. In other words, under this representation, for each realization~$\omega\in \Omega$, the $n$-fold iterate~$\mathbf{f}^{n}(x,\omega)$ gives the state of the system at time~$n$ with initial condition~$x$. The evolution of such a system is shown in Fig.\ref{fig:bilayer}.
\begin{figure} [hbt]
\begin{center}
\includegraphics[scale= 0.7]{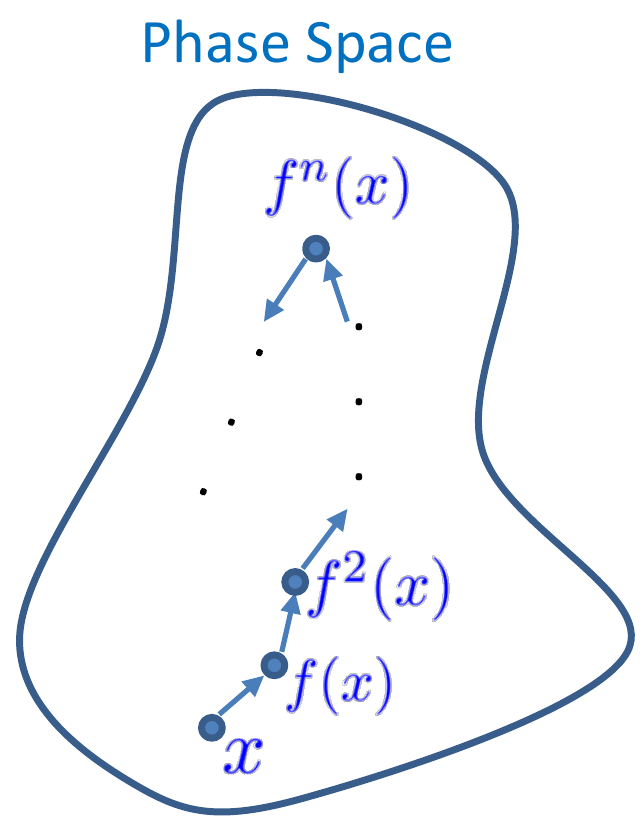}
\caption{Evolution of a realization of the stochastic dynamical system given by the iterates of~$\mathbf{f}$.}\label{fig:bilayer}
\end{center}
\end{figure}
We study the evolution of ties under this dynamical system. We are concerned with the long term behavior of the dynamical system~$\left(\mathbf{p}(n)\right)$~\eqref{eq:stoc22}. This leads to the notion of strong-fixed point and strong-attractor that we introduce next.
\begin{definition}[Strong-fixed point]
Let~$\mathbf{F}\,:\,\mathcal{D}\times \Omega \rightarrow \mathcal{D}$ be a discrete-time random dynamical system. We call~$x^*\in\mathcal{D}$ a strong-fixed point of~$\mathbf{F}$, whenever
\begin{equation}
x^{*}=\mathbf{F}(x^{*},\omega),\nonumber
\end{equation}
for almost all~$\omega\in \Omega$.
\end{definition}
Let
\begin{equation}
\mathcal{E}_{\mathbf{F}}=\left\{x\in \mathcal{D}\,:\, x=\mathbf{F}(x,\omega),\,\mbox{for almost all }\omega\right\}\nonumber
\end{equation}
be the set of strong-fixed points of the stochastic dynamical system~$\mathbf{F}$. From now on, we drop the subindex~$\mathbf{F}$ to write~$\mathcal{E}$ instead of~$\mathcal{E}_{\mathbf{F}}$, as the underlying random vector field is assumed to be the one in~\eqref{eq:stoc22}.


A goal in this paper is to prove that, under certain conditions on the update rules~$\mathbf{f}_{+}$,~$\mathbf{f}_{-}$,~$\mathbf{g}_{-}$, and the conditional distributions of~$\mathbf{R}$ and~$\mathbf{C}$, a finite proper subset of the set of binary matrices
\begin{equation}
\mathcal{E}=\left\{E_1,E_2,\ldots,E_M\right\}\subset \left\{0,1\right\}^{N \times N}\nonumber
\end{equation}
is a strong-attractor to the dynamical system~\eqref{eq:stoc22}, i.e.,
\begin{equation}
d\left(\mathbf{p}(n),\mathcal{E}\right)\overset{a.s.}\longrightarrow 0,\nonumber
\end{equation}
where~$d$ is the Euclidean metric. Refer to Theorem~\ref{th:final} (or Corollary~\ref{th:dual}) or its more general version (with sparse in time perturbations) in Theorem~\ref{th:finalwexplor}. In other words, given a model in the family~\eqref{eq:stoc22}, a \textbf{deterministic network} or graph emerges in the long-run almost surely. Technically speaking, as will be illustrated later, standard convergence arguments for Markov processes are not directly applicable to the class of dynamical systems given by~\eqref{eq:stoc22} and we develop new techniques that might be of independent interest. For instance, we observe that standard absorbing-like arguments for Markov processes (e.g., irreducibility plus ergodicity) do not apply to our setup. In fact, if the process~$\left(\mathbf{p}(n)\right)$ departs from the interior~$\left(0,1\right)^{N\times N}$, and the reinforcement~$\mathbf{f}_{+}$ and penalization laws~$\mathbf{g}_{-}$,~$\mathbf{f}_{-}$ are \emph{soft}, i.e.,~$\mathbf{0}_{N}<\mathbf{g}_{-}(\mathbf{x}),\,\mathbf{f}_{-}(\mathbf{x}),\,\mathbf{f}_{+}(\mathbf{x})<\mathbf{1}_{N}$ for all~$\mathbf{x}\in \left(0,1\right)^N$ (e.g.,~\eqref{eq:soft}), then~$\left(\mathbf{p}(n)\right)$ never coalesces with extreme points~$\mathcal{E}\subset \left\{0,1\right\}^{N \times N}$ in finite time; rather, it accumulates onto it. We prove such convergence results by showing the following: \textbf{i)} any~$E_{\ell}\in \mathcal{E}$ is a local attractor with positive probability (bounded away from zero in a small neighborhood of~$E_{\ell}$) \textbf{(}\textbf{Theorem}~\ref{co:bounde}\textbf{)}; \textbf{ii)} any arbitrarily small cover of~$\mathcal{E}$ is recurrent \textbf{(}\textbf{Theorem}~\ref{th:ball}\textbf{)}. These two assertions will imply the strong-convergence via Theorem~\ref{th:markovhomog}.

\emph{Remark:} As discussed in the third-to-last paragraph in Section~\ref{sec:probform}, if the capacity of response of all nodes is unlimited, then the dynamical system~$\left(\mathbf{p}(n)\right)$ is decoupled. In this case, the convergence referred to above follows as a corollary to the Doob's (sub-)martingale convergence theorem~\cite{williamsprobability,Diffusion}, as each coordinate~$\left(p_{ij}(n)\right)$ evolves independently as a sub-martingale or a super-martingale -- depending on whether the reinforcement rule~$f_{i,+}$ is stronger than the fading~$g_{i,-}$ or vice-versa. When limited capacity is assumed, the coordinates are coupled and the process loses the martingale structure. Doob's convergence does not apply to the resulting coupled stochastic dynamical system~$\left(\mathbf{p}(n)\right)$.

\section{0-1 Convergence Law}\label{sec:convresults}

%
%

In this section, we establish the~$0-1$ law in Theorems~\ref{th:final}-\ref{th:dual}-\ref{th:finalwexplor} for the stochastic dynamical system~$\left(\mathbf{p}(n)\right)$ described in~\eqref{eq:stoc22}. First, we define
\begin{eqnarray}
\mathcal{C}_{+} & \overset{\Delta}= & \left\{f\in\left[0,1\right]^{\left[0,1\right]}\,:\,f>{\sf id}\,\mbox{ over }\left.\left[0,1\right.\right),\,f(1)=1,\,f\,\,\mbox{is increasing}\right\}\nonumber\\
\mathcal{C}_{-} & \overset{\Delta}= & \left\{f\in\left[0,1\right]^{\left[0,1\right]}\,:\,f<{\sf id}\,\mbox{ over }\left.\left(0,1\right.\right],\,f(0)=0,\,f\,\,\mbox{is increasing}\right\}\nonumber\\
\mathcal{P}_{+} & \overset{\Delta}= & \left\{f\in\left[0,1\right]^{\left[0,1\right]}\,:\,\sum_{n}\left(1-f^{n}(x)\right)<\infty\right\}\nonumber\\
\mathcal{P}_{-} & \overset{\Delta}= & \left\{f\in\left[0,1\right]^{\left[0,1\right]}\,:\,\sum_{n}f^{n}(x)<\infty\right\},\nonumber
\end{eqnarray}
where~${\sf id}\,:\,\left[0,1\right]\rightarrow \left[0,1\right]$ is the identity map~${\sf id}(x)=x$. For instance, for
\begin{equation}
f_{+}(x)=\frac{1+x}{2},\,\,g_{-}(x)=\frac{x}{3},\label{eq:soft}
\end{equation}
we have~$f_{+}\,:\,\left[0,1\right]\to\left[0,1\right]\in \mathcal{C}_{+}\cap \mathcal{P}_{+}$ and~$g_{-} \,:\,\left[0,1\right]\to\left[0,1\right]\in \mathcal{C}_{-}\cap \mathcal{P}_{-}$.

Next, we summarize all the assumptions on the dynamical system~$\left(\mathbf{p}(n)\right)$ to prove the main Theorems~\ref{th:final}-\ref{th:dual}-\ref{th:finalwexplor}.

$\bullet$ \textbf{(RPF laws)} $f_{i,+}\in\mathcal{C}_{+}\cap \mathcal{P}_{+}$,~$g_{i,-}\in\mathcal{C}_{-}\cap \mathcal{P}_{-}$ and~$f_{i,-}\in\mathcal{C}_{-}$ for all nodes~$i$\footnote{Note that we do not assume that the $n$-fold iterates of the penalization rules~$f_{i,-}$ are Cesaro summable, i.e.,~$f_{i,-}\in\mathcal{P}_{-}$.};

$\bullet$ \textbf{(Responses)} We assume that nodes have limited capacity and tend to utilize all of their bandwidth for response. Specifically: i) each node~$i$ can only respond to at most~$Q_{i}$ calls at each time~$n$; ii) if the number of incoming calls to~$i$ is below its capacity~$Q_{i}$, then~$i$ responds to all callers. Mathematically,
\begin{equation}
R(n)\mathbf{1}_{N}=\min\left\{\left(C^{\top}(n)\mathbf{1}_{N}\right)^{\top},\, \mathbf{Q}\right\},\label{eq:responsecapa}
\end{equation}
where~$\mathbf{Q}=\left(Q_1,\ldots,Q_N\right)$ is the vector of capacities. We also assume that~$\left(R(n)\right)$ is memoryless. The~$\min$ in~\eqref{eq:responsecapa} is taken entry-wise. Equivalently, in terms of the conditional distribution of~$\left(R(n)\right)$, and as discussed earlier (see~\eqref{eq:lasteq0}-\eqref{eq:Mequations}), a particular response policy~$\left(R(n)\right)$ is characterized by a specific choice in the class~$\mathcal{M}_1\times \ldots \times \mathcal{M}_N$. As noted before, the criterion (biased or not) to select which set of~$Q_i$ nodes to respond to, in case~$i$ receives more calls than the capacity~$Q_i$ is not relevant for our convergence theorems, and thus the analytical results presented hold for a broad class of models obeying the stochastic dynamical system~\eqref{eq:stoc22};

$\bullet$ \textbf{(Calls)} in Subsection~\ref{subsec:nopert}, we assume~$\mathbb{P}\left(C_{ij}(n)=1\left|\mathbf{p}(n)\right.\right)=p_{ij}(n)$, and in Subsection~\ref{subsec:pert} we consider (fading) perturbations on the conditional distribution of~$\mathbf{C}(n)$ as described later in~\eqref{eq:pertlaw}. In both cases, we exclude self-calls, that is,~$\mathbb{P}\left(C_{ii}(n)=1\left|\mathbf{p}(n)\right.\right)=0$.

%

Under the above assumptions, it can be verified that the subset of binary matrices
\begin{equation}
\mathcal{E}=\left\{E\in\left\{0,1\right\}^{N\times N}\,:\, {\sf tr}\left(E\right)=0,\,\mathbf{1}^{\top} E\leq \left[Q_1\,\cdots\,Q_{N}\right]\right\}=:\left\{E_{1},\ldots,E_{M}\right\}\label{eq:setimport}
\end{equation}
is the set of strong-fixed points of the stochastic dynamical system~\eqref{eq:stoc22}, where~$Q_{i}$ is the capacity of response to node~$i$. In this section, we show that~$\mathcal{E}$ is in fact a global strong-attractor to~\eqref{eq:stoc22}, i.e., the process~$\left(\mathbf{p}(n)\right)$ converges almost surely to a matrix in~$\mathcal{E}$ from arbitrary initial condition~$\mathbf{p}(0)$.

\subsection{Main Convergence Result}\label{subsec:nopert}

The following lemma and corollary are crucial to what follows.
\begin{lemma}\label{th:sequence}
Let~$p_n\in\left.\left[0,1\right.\right)$ for all~$n\in\mathbb{N}$. Then,
\begin{equation}
\sum_{n} (1-p_n)< \infty \Rightarrow \Pi_n p_n >0.\nonumber
\end{equation}
\end{lemma}
The sequence~$\left(\Pi_k^n p_k\right)_n$ necessarily converges as it is monotonic. Also, if~$p_n\leq q$ for all~$n$ (or at least for infinitely many~$n$) for some~$q<1$, then~$\Pi_n p_n=0$. Therefore, a necessary condition to have~$\Pi_n p_n>0$ is that~$p_n$ converges to~$1$. Lemma~\ref{th:sequence} provides a sufficient condition: $p_n$ shall converge to~$1$ fast enough.

\begin{proof}
The result is adapted from~\cite{williamsprobability}, but since it is left as an exercise in that reference, we provide our own proof. Define~$\widetilde{p}_n:=1-p_n$. To start with, assume that~$\sum_{n}\widetilde{p}_n< 1$. Expand the product~$\Pi_n (1-p_n)$ to observe that
\begin{equation}
\Pi_n (1-\widetilde{p}_n) \geq 1- \sum_n \widetilde{p}_n>0.\nonumber
\end{equation}
Now, assume that~$1<\sum_n \widetilde{p}_n=m<\infty$, then
\begin{equation}
\exists{\overline{K}}:\,\,\sum_{\ell=1}^k \widetilde{p}_{\ell}\in\left.\left(m-1,m\right.\right],\,\,\,\,\,\forall{k\geq \overline{K}}.\nonumber
\end{equation}
Now,
\begin{equation}
\Pi_n (1-\widetilde{p}_n)= \Pi_{n=1}^{\overline{K}} (1-\widetilde{p}_n) \times \Pi_{n\geq \overline{K}+1} (1-\widetilde{p}_n),\nonumber
\end{equation}
and the result follows.
\end{proof}

\begin{corollary}\label{th:corosequence}
Let~$f\,:\,\left[0,1\right]\rightarrow \left[0,1\right]$ be an increasing function. Assume that
\begin{equation}
\sum_{n} (1-f^n(\widetilde{p}))< \infty,\label{eq:cond}
\end{equation}
for some~$\widetilde{p}\in\left(0,1\right)$. Then, there exists~$1>\epsilon>0$, such that
\begin{equation}
p\in\left.\left[\widetilde{p},1\right.\right)\Rightarrow \Pi_n f^{n}(p)>\epsilon,\nonumber
\end{equation}
i.e.,~$\epsilon$ does not depend on the choice of~$p\in\left.\left[\widetilde{p},1\right.\right)$ (though it may depend on~$\widetilde{p}$).
\end{corollary}

Corollary~\ref{th:corosequence} provides a uniform boundness property for the infinite product~$\Pi_n f^{n}(p)$ over some subinterval~$\left.\left[\widetilde{p},1\right.\right)$, which will be relevant latter.

\begin{proof}
Observe that~$\Pi_n f^{n}(\widetilde{p})>0$ or better~$\Pi_n f^{n}(\widetilde{p})>\epsilon>0$, for some~$\epsilon$, due to~\eqref{eq:cond} and Lemma~\ref{th:sequence}. If~$p\geq \widetilde{p}$, then
\begin{equation}
f^{n}(p)\geq f^{n}(\widetilde{p}),\,\,\,\,\forall{n},\nonumber
\end{equation}
as~$f$ is increasing, and thus,
\begin{equation}
\Pi_n f^{n}(p)\geq \Pi_n f^{n}(\widetilde{p})>\epsilon,\nonumber
\end{equation}
for some~$\epsilon>0$ and all~$p\geq \widetilde{p}$.
\end{proof}

In what follows, it is relevant to recall that, to each binary matrix $E_{\ell}\in \mathcal{E}$, there exists an underlying support graph~$G_{\ell}=\left(V,\mathcal{D}_{\ell}\right)$ whose set of edges is given by
\begin{equation}
\mathcal{D}_{\ell}=\left\{(i,j)\in V \times V\,:\, E_{\ell}(i,j)=1\right\}.\label{eq:onetoone}
\end{equation}
Now, define the family of events (indexed by~$\mathcal{D}_{\ell}$)
\begin{equation}
\begin{array}{ccc}
\mathcal{A}_{\mathcal{D}_{\ell}}(n) & = & \left\{\omega\in \Omega\,:\,C_{ij}(n)=1,\,R_{ji}(n)=1 \mbox{ for all } (i,j)\in \mathcal{D}_{\ell}\right\}\\
& & \cap \left\{\omega\in \Omega,\,:\,C_{ij}(n)=0,\,\mbox{for all } (i,j)\notin \mathcal{D}_{\ell}\right\}
\end{array}
\end{equation}
i.e., the realizations~$\omega\in \Omega$ where, at time~$n$, for~$(i,j)\in\mathcal{D}_{\ell}$,~$i$ calls~$j$ and~$j$ responds to~$i$, and~$i$ does not call~$j$ for all the remaining edges~$\left(i,j\right)\notin \mathcal{D}_{\ell}$. The set~$\mathcal{A}_{\mathcal{D}_{\ell}}(n)$ is of positive probability if and only if~$p_{ij}(n)>0$ for all~$(i,j)\in \mathcal{D}_{\ell}$;~$p_{ij}(n)<1$ for all~$(i,j)\notin \mathcal{D}_{\ell}$; and~$E_{\ell}\in \mathcal{E}$. Recall that~$\mathcal{D}_{\ell}$ and~$E_{\ell}$ are in one-to-one correspondence given by~\eqref{eq:onetoone}.

Define
\begin{equation}
\overline{B}_{\delta}\left(\widetilde{\mathbf{p}}\right)\overset{\Delta}=\left\{\mathbf{p}\in \left[0,1\right]^{N \times N}\,:\, \left|p_{ij}-\widetilde{p}_{ij}\right|\leq \delta,\,\mbox{for all}\,\,1\leq i,j\leq N\right\}
\end{equation}
as the closed $L_{\infty}$ ball centered at the probability matrix~$\widetilde{\mathbf{p}}$ and of radius~$\delta$ restricted to the set~$\left[0,1\right]^{N \times N}$.

The next theorem states that any~$E_{\ell}\in \mathcal{E}$ (recall~\eqref{eq:setimport}) is a \emph{local} attractor with strictly positive probability in a small enough neighborhood of~$E_{\ell}$, for the stochastic dynamical system~$\left(\mathbf{p}(n)\right)$.
%
\begin{theorem}\label{co:bounde}
There exists~$\epsilon>0$ such that
\begin{equation}
\mathbb{P}\left(\mathbf{p}(n)\overset{n\to \infty}\longrightarrow E_{\ell}\left| \mathbf{p}(0)\in \overline{B}_{\delta}\left(E_{\ell}\right)\right.\right)> \epsilon,\nonumber
\end{equation}
for some~$\delta>0$ small enough and for all~$E_{\ell}\in\mathcal{E}$.
\end{theorem}

\begin{proof}
Note that,
\begin{eqnarray}
\mathbb{P}\left(\mathbf{p}(n)\to E_{\ell}\left| \mathbf{p}(0)\in \overline{B}_{\delta}\left(E_{\ell}\right)\right.\right) & \geq & \mathbb{P}\left( \mathcal{A}_{\mathcal{D}_{\ell}}(n) \forall_{n\geq 0}\left|\mathbf{p}(0)\in \overline{B}_{\delta}\left(E_{\ell}\right)\right.\right)\nonumber\\
& = & \lim_{n} \mathbb{P}\left(\mathcal{A}_{\mathcal{D}_{\ell}}(k) \forall_{k\leq n}\left| \mathbf{p}(0)\in \overline{B}_{\delta}\left(E_{\ell}\right)\right.\right)\nonumber\\
& = & \lim_{n} \mathbb{P}\left(\mathcal{A}_{\mathcal{D}_{\ell}}(n) \left| \mathbf{p}(0)\in \overline{B}_{\delta}\left(E_{\ell}\right),\,\mathcal{A}_{\mathcal{D}_{\ell}}(n-1),\ldots,\mathcal{A}_{\mathcal{D}_{\ell}}(0)\right.\right)\nonumber\\
& & \times \mathbb{P}\left(\mathcal{A}_{\mathcal{D}_{\ell}}(n-1)\left| \mathbf{p}(0)\in \overline{B}_{\delta}\left(E_{\ell}\right),\,\mathcal{A}_{\mathcal{D}_{\ell}}(n-2),\ldots,\mathcal{A}_{\mathcal{D}_{\ell}}(0)\right. \right)\nonumber\\
& & \times \ldots \times \mathbb{P}\left(\mathcal{A}_{\mathcal{D}_{\ell}}(0) \left| \mathbf{p}(0)\in \overline{B}_{\delta}\left(E_{\ell}\right)\right.\right)\nonumber\\
& = & \Pi_{n=0}^{\infty} \Pi_{ij \in \mathcal{D}_{\ell}} f_{i,+}^{n}\left(p_{ij}(0)\right)\Pi_{ij \notin \mathcal{D}_{\ell}} \left(1-g_{i,-}^{n}\left(p_{ij}(0)\right)\right),\label{eq:condiconv}
\end{eqnarray}
where the last equality holds from the characterization of~$\mathcal{A}_{\mathcal{D}_{\ell}}$ and the assumption that~$\mathbf{R}$ fulfills~\eqref{eq:responsecapa}. Since we are conditioning on the event~$\mathbf{p}(0)\in \overline{B}_{\delta}(E_{\ell})$, we have in~\eqref{eq:condiconv} that~$p_{ij}(0)\in \left.\left(1-\delta,1\right.\right]$ for~$(i,j) \in \mathcal{D}_{\ell}$ and~$p_{ij}(0)\in \left.\left(0,\delta\right.\right]$, otherwise. Therefore, from Corollary~\ref{th:corosequence} and the assumption that~$f_{i,+}\in\mathcal{C}_{+}\cap \mathcal{P}_{+}$,~$g_{i,-}\in\mathcal{C}_{-}\cap \mathcal{P}_{-}$, we have
\begin{eqnarray}
\mathbb{P}\left(\mathbf{p}(n)\to E_{\ell}\left| \mathbf{p}(0)\in \overline{B}_{\delta}\left(E_{\ell}\right)\right.\right) & \geq & \Pi_{ij \in \mathcal{D}_{\ell}} \Pi_{n=0}^{\infty} f_{i,+}^{n}\left(p_{ij}(0)\right) \times \Pi_{ij \notin \mathcal{D}_{\ell}} \Pi_{n=0}^{\infty} \left(1-g_{i,-}^{n}\left(p_{ij}(0)\right)\right)\nonumber\\
& > & \epsilon,\nonumber
\end{eqnarray}
for some~$\epsilon$ and for all~$\ell$.
\end{proof}
To establish the main result of this subsection, Theorem~\ref{th:final}, or its dual formulation, Corollary~\ref{th:dual}, it will be crucial to show that the hitting time~$T$ to~$B:=\bigcup_{\ell}\overline{B}_{\delta}(E_{\ell})$ is integrable -- hence, almost surely finite -- i.e.,~$B$ is recurrent. For this, the next theorem will be useful.

\begin{theorem}[Lemma in Section $10.11$ in~\cite{williamsprobability}]\label{th:hitting}
Let~$T$ be a hitting time with respect to a filtration~$\left(\mathcal{F}_n\right)$. If for some~$\epsilon>0$ and some~$N\in \mathbb{N}$ we have
\begin{equation}
\mathbb{P}\left(T\leq n+N \left.\right| \mathcal{F}_n\right)> \epsilon,\,\,\,\forall{n\in\mathbb{N}},\label{eq:assumption}
\end{equation}
then,
$E\left[T\right]<\infty$.
\end{theorem}

\begin{proof}
This theorem is adapted from~\cite{williamsprobability}, but since it is left there as an exercise, we prove it here for completion. We will prove that, as a consequence of assumption~\eqref{eq:assumption},
\begin{equation}
\mathbb{P}\left(T> kN\right) \leq \left(1-\epsilon\right)^{k},\label{eq:ind}
\end{equation}
for all~$k$. Note that this will conclude the proof of this theorem as
\begin{equation}
E\left(\frac{T}{N}\right)=\sum_{k=1}^{\infty} \mathbb{P}\left(\frac{T}{N}>k\right)=\frac{\epsilon}{1-\epsilon}<\infty,\nonumber
\end{equation}
and thus,~$E(T)<\infty$. We prove~\eqref{eq:ind} by induction on~$k$. For~$k=0$, the assertion is clear. Now, assume it is true for an integer~$k$. Then,
\begin{eqnarray}
\mathbb{P}\left(T>(k+1)N\right) & = & \mathbb{P}\left(T>(k+1)N, T>kN\right) \nonumber\\
& = & E\left(E\left(\mathbf{1}_{T>(k+1)N}\mathbf{1}_{T>kN}\left|\right.\mathcal{F}_{kN}\right)\right) \nonumber\\
& = & E\left(\mathbf{1}_{T>kN}\mathbb{P}\left(T>kN+N\left|\right. \mathcal{F}_{kN}\right)\right) \nonumber\\
& \leq & (1-\epsilon)E\left(\mathbf{1}_{T>kN}\right)=(1-\epsilon)\mathbb{P}\left(T>kN\right) \nonumber\\
& \leq & (1-\epsilon)^{k}(1-\epsilon)=(1-\epsilon)^{k+1},\nonumber
\end{eqnarray}
where the first inequality follows from~\eqref{eq:assumption}. This concludes the proof.
\end{proof}

The next theorem, Theorem~\ref{th:ball}, states that any arbitrarily \emph{small} (with non-empty interior) cover $\bigcup_{\ell}\overline{B}_{\delta}\left(E_{\ell}\right)$ to~$\mathcal{E}$ is a recurrent set.
\begin{theorem}\label{th:ball}
We have that
\begin{equation}
\mathbf{p}(n)\in B:=\bigcup_{k=1}^M \overline{B}_{\delta}\left(E_{k}\right)\,\,\,\mbox{infinitely often},\nonumber
\end{equation}
for all~$\delta>0$.
\end{theorem}

\begin{proof}
Let~$\widetilde{\alpha}:= \max\left\{i=1,\ldots,N\,:\,f_{i,-}(1),g_{i,-}(1)\right\}$, and define the set
\begin{equation}
\mathcal{L}\overset{\Delta}=\left\{\mathbf{p}\in \left[0,1\right]^{N \times N}\,:\,\exists{E \in {\sf Per}}\,\,E\mathbf{p}_i\leq \left[\mathbf{1}_{Q_i}^{\top} \,\, \widetilde{\alpha}\mathbf{1}^{\top}_{N-Q_i}\right]^{\top} \right\},
\end{equation}
where~$\mathbf{p}_i$ is the $i$-th column of the probability matrix~$\mathbf{p}$, and~${\sf Per}$ is the set of~$N \times N$ permutation matrices. Due to the limited capacity of response~$Q_{i}$ of each node~$i$, we observe that~$\mathbf{p}(n)\in \mathcal{L}$ for all~$n\geq 1$. Let~$\alpha=\max\left\{1/2,\widetilde{\alpha}\right\}$ and observe that
\begin{equation}
\mathcal{L}=\bigcup_{\ell=1}^{M} \overline{B}_{\alpha}\left(E_{\ell}\right).\nonumber
\end{equation}
Consider its disjointification~$\left(K_{\ell}\right)_{\ell}$
\begin{equation}
K_m = \overline{B}_{\alpha}\left(E_{m}\right)\setminus \bigcup_{k<m} \overline{B}_{\alpha}\left(E_{k}\right).\nonumber
\end{equation}
Now, define
\begin{equation}
k_{i,-}:=\min\left\{n\in\mathbb{N}\,:\,g_{i,-}^{n}(1) \in \left.\left[0,\delta\right.\right)\right\}\nonumber
\end{equation}
as the number of iterates to go from~$1$ to~$\left.\left[0,\delta\right.\right)$ by n-fold iterates of~$g_{i,-}$.
Also, define
\begin{equation}
k_{i,+}:=\min\left\{n\in\mathbb{N}\,:\,f_{i,+}^{n}(0) \in \left.\left[1-\delta,\delta\right.\right)\right\}\nonumber
\end{equation}
to be the number of iterates to go from~$0$ to~$\left.\left[\delta-1,1\right.\right)$ by iterates of~$f_{i,+}$. Choose~
\begin{equation}
N\geq \max\left\{i=1,\ldots,N\,:\,k_{i,+},k_{i,-}\right\}.\nonumber
\end{equation}
Let
\begin{equation}
T:=\min\left\{n\in\mathbb{N}\,:\,\mathbf{p}(n)\in B\right\}\nonumber
\end{equation}
be the hitting time to the set~$B$ by the process~$\left(\mathbf{p}(n)\right)$. For~$n\geq 1$ we have
\begin{eqnarray}
\mathbb{P}\left(T\leq n+N\left.\right| \mathcal{F}_n\right) & = & \mathbb{P}\left(T\leq n+N,\,\mathbf{p}(n)\in \mathcal{L}\left.\right| \mathcal{F}_n\right)\nonumber\\
& = & \sum_{m} \mathbb{P}\left(T\leq n+N\left.\right| \mathcal{F}_n,\,\mathbf{p}(n)\in K_m\right)\mathbb{P}\left(\mathbf{p}(n)\in K_m\right).\nonumber
\end{eqnarray}
Now,
\begin{eqnarray}
 \mathbb{P}\left(T\leq n+N\left.\right| \mathcal{F}_n,\,\mathbf{p}(n)\in K_m\right) & \geq & \mathbb{P}\left( \mathcal{A}_{\mathcal{D}_{m}}(k) \,\,\forall_{n \leq k\leq N+n}\left|\right.\mathcal{F}_n,\,\mathbf{p}(n)\in K_m\right)\nonumber\\
& \geq & \Pi_{k\geq n}^{N+n} \Pi_{ij \in \mathcal{D}_{\ell}} f_{i,+}^{k}(1-\alpha)\Pi_{ij \notin \mathcal{D}_{\ell}} \left(1-g_{i,-}^{k}(\alpha)\right)\label{eq:condiconv2}\\
& := & \epsilon >0,
\end{eqnarray}
where the first inequality follows from the choice of~$N$, and the second inequality follows similarly to as done in~\eqref{eq:condiconv}. Due to the monotonicity of~$f_{i,+}$ and~$g_{i,-}$, all the terms in the above finite product are positive, and, thus,~$\epsilon>0$ and it does not depend on~$n$. The theorem now follows from Theorem~\ref{th:hitting}.
\end{proof}

The next theorem allows us to combine Theorem~\ref{co:bounde} and Theorem~\ref{th:ball} to establish the main Theorem~\ref{th:final}.
\begin{theorem}\label{th:markovhomog}
Let~$p(n)$ be a time-homogeneous Markov process on a set~$X$. Let~$E\subset X$. Then,
\begin{equation}
\left\{\begin{array}{ccc} \mathbb{P}\left(p(n)\in E,\,\mbox{infinitely often}\right) & = & 1\\
\mathbb{P}\left(p(n)\in E,\,\forall{n\geq 0}\left|\right. p(0)\in E\right) & > & \epsilon\end{array}\right.\Longrightarrow \mathbb{P}\left(p(n)\in E,\,\mbox{eventually}\right)=1.\nonumber
\end{equation}
In other words, if~$E$ is recurrent and invariant with positive probability, then the tail of~$\left(p(n)\right)$ lies in~$E$.
\end{theorem}

\begin{proof}
Define the sequence of stopping times~$\left(T_i\right)_{i=1}^{\infty}$, where~$T_1$ is the time of first return to the set~$E$; $T_N$ is the time of the~$N$-th return to the set~$E$. Note that~$T_1< T_2 < T_3 <\ldots$. Define~$U$ as the random variable associated with the total number of exits from the set~$E$ over all time~$\mathbb{N}$. Note that if~$\left(p(n)\right)$ exits~$E$, it will return in finite time as~$E$ is recurrent. In particular, we have the equality of the events
\begin{equation}
\left\{U\geq N\right\}=\left\{T_N <\infty\right\},\nonumber
\end{equation}
for all~$N\geq 1$. We show by induction on~$N$ that
\begin{equation}
\mathbb{P}\left(U\geq N \left|\right. p(0)\in E\right) \leq (1-\epsilon)^N,\nonumber
\end{equation}
for all~$N\geq 1$. Indeed, for~$N=1$
\begin{equation}
\mathbb{P}\left(U\geq 1 \left|\right. p(0)\in E\right) = 1-\mathbb{P}\left(p(n)\in E,\,\forall{n\geq 0}\left|\right. p(0)\in E\right) \leq 1-\epsilon.\nonumber
\end{equation}
Now,
\begin{eqnarray}
\mathbb{P}\left(U\geq N \left|\right. p(0)\in E\right) & = & \mathbb{P}\left(T_N< \infty \left|\right. p(0)\in E\right)\nonumber\\
& = & \mathbb{P}\left(\left|T_{N}-T_{N-1}\right|<\infty,\,T_{N-1}< \infty \left|\right. p(0)\in E\right)\nonumber\\
& = & \mathbb{P}\left(\left|T_{N}-T_{N-1}\right|<\infty \left|\right. p(0)\in E,\,T_{N-1}< \infty\right)\mathbb{P}\left(T_{N-1}< \infty \left|\right. p(0)\in E\right)\nonumber\\
& = & \mathbb{P}\left(\left|T_{N}-T_{N-1}\right|<\infty \left|\right.T_{N-1}< \infty\right)\mathbb{P}\left(T_{N-1}< \infty \left|\right. p(0)\in E\right)\nonumber\\
& = & \mathbb{P}\left(U\geq 1 \left|\right. p(0)\in E\right)\mathbb{P}\left(U\geq N-1\left|\right.p(0)\in E\right)\nonumber\\
& \leq & \left(1-\epsilon\right)\left(1-\epsilon\right)^{N-1}=\left(1-\epsilon\right)^N,
\end{eqnarray}
where the fifth equality comes from the strong Markov property, and the last inequality results from the induction hypothesis. Thus,
\begin{equation}
E\left(U\left|\right. p(0)\in E\right)= \sum_N \mathbb{P}\left(U\geq N \left|\right. p(0)\in E\right) < \infty,\label{eq:strongmar}
\end{equation}
and therefore,
\begin{equation}
\mathbb{P}\left(U=\infty\left|\right. p(0)\in E\right)=0.\nonumber
\end{equation}
Define~$T$ as the stopping time associated with the first time the Markov process~$\left(p(n)\right)$ is in~$E$ -- note that, in general,~$T$ is different from~$T_1$ as the process can start in~$E$, and in this case~$T=0$, but~$T_1>0$. Since~$E$ is recurrent, we have that~$T<\infty$ almost surely and
\begin{equation}
E\left(U\right)=E\left(U \left|\right. p(T)\in E\right)\mathbb{P}\left(p(T)\in E\right)= E\left(U \left|\right. p(T)\in E\right)<\infty,\nonumber
\end{equation}
where the last inequality follows from~\eqref{eq:strongmar}, the homogeneity of~$\left(p(n)\right)$, and the strong Markov property. Hence, we conclude that
\begin{equation}
\mathbb{P}\left(p(n)\in E,\,\mbox{eventually}\right)=1-\mathbb{P}\left(\left\{p(n)\in E,\,\mbox{eventually}\right\}^{c}\right)=1-\mathbb{P}\left(U=\infty\right)=1.\nonumber
\end{equation}
\end{proof}

We are now prepared to obtain Theorem~\ref{th:final}.

\begin{theorem}\label{th:final}
Let~$\left(\mathbf{p}(n)\right)$ be solution to the stochastic dynamical system~\eqref{eq:stoc22}, with~$\mathbf{C}(n)$, and $\left(\mathbf{R}(n)\right)$ satisfying the assumptions presented in the beginning of this section, with~$f_{i,+}\in\mathcal{C}_{+}\cap \mathcal{P}_{+}$,~$g_{i,-}\in\mathcal{C}_{-}\cap \mathcal{P}_{-}$ and~$f_{i,-}\in\mathcal{C}_{-}$. Then,
\begin{equation}
\mathbb{P}\left(\mathbf{p}(n)\overset{n\rightarrow \infty}\longrightarrow E_{\ell},\mbox{ for some }\ell\right)=1.
\end{equation}
\end{theorem}
The proof follows from Theorems~\ref{co:bounde}-\ref{th:ball} (via Theorem~\ref{th:markovhomog}) as suggested by Fig.~\ref{fig:ideaofproof}.
\begin{figure}
  \centering
  \includegraphics[width=.4\linewidth]{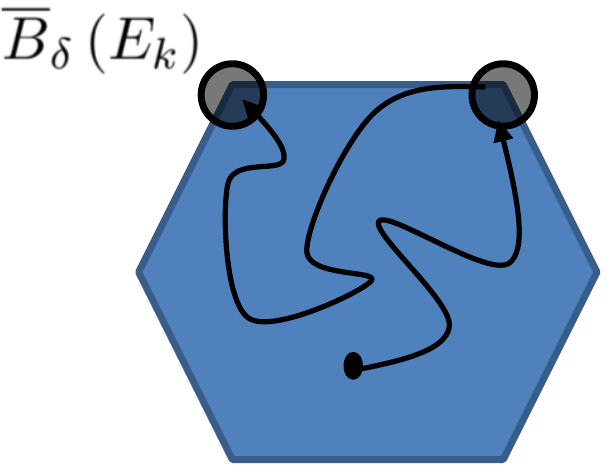}
  \caption{Summary of the proof to Theorem~\ref{th:final}: i) the extreme matrices~$E_{k}$ are local attractors (Theorem~\ref{co:bounde}); ii) any small neighborhood covering~$\mathcal{E}$ is recurrent (Theorem~\ref{th:ball}).}
  \label{fig:ideaofproof}
\end{figure}

Note that,~$\left(\mathbf{p}(n)\right)$ is a stochastic process whose limiting behavior is given by a random variable~$\mathbf{p}(\infty)$ (with support on~$\mathcal{E}$). The initial condition~$\mathbf{p}(0)$ is also a random variable, and therefore, the proper setting to characterize the dynamics in terms of the attractors and the corresponding basin of attraction is the dual space of probability measures on~$\left[0,1\right]^{N \times N}$, as established in Corollary~\ref{th:dual}.
\begin{corollary}\label{th:dual}
Let~$\mu_0$ be a probability measure on~$\left[0,1\right]^{N \times N}$ associated with the distribution of~$\mathbf{p}(0)$ and let~$\mathbf{p}(n)\overset{d}\sim\mu_n$. Then, under the same assumptions of Theorem~\ref{th:final}, there exists~$\eta \in {\sf co}\left\{\delta_{E_{1}},\ldots,\delta_{E_{M}}\right\}$ such that
\begin{equation}
\mu_{n}\Rightarrow \eta,\nonumber
\end{equation}
where~$\Rightarrow$ means convergence with respect to the weak-$\star$ topology on the (compact) space of probability measures on~$\left[0,1\right]^{N \times N}$ (refer to~\cite{billi}); and~${\sf co}$ stands for the convex hull of a set (refer to~\cite{convex}).
\end{corollary}

\begin{proof}
Consider the partition~$\bigcup_{i=1}^M \Omega_i$ where
\begin{equation}
\Omega_i\overset{\Delta}= \left\{\omega\in\Omega\,:\, \mathbf{p}(n,\omega)\rightarrow E_i\right\}.\nonumber
\end{equation}
Now, let~$\phi\,:\, \left[0,1\right]^{N \times N} \to \mathbb{R}$ be any continuous bounded function and note that
\begin{eqnarray}
\int \phi \,\,d\mu_n & = & \int \phi \,\,d\mathbb{P}_{\star \mathbf{p}(n)}= \int \phi(\mathbf{p}(n,\omega))\,\, d\mathbb{P}(\omega)\nonumber\\
& = & \sum_{i=1}^M \int_{\Omega_i} \phi(\mathbf{p}(n,\omega))\,\, d\mathbb{P}  \overset{n\rightarrow \infty}\longrightarrow \sum_{i=1}^M \int_{\Omega_i} \phi(E_i) \,\,d\mathbb{P}= \sum_{i=1}^M \phi(E_i)\mathbb{P}\left(\Omega_i\right)\nonumber\\
& = & \sum_{i=1}^M \int \phi \,\,d\delta_{E_i} \mathbb{P}\left(\Omega_i\right) =  \int \phi \,\,d\left(\sum_{i=1}^M \mathbb{P}\left(\Omega_i\right) \,\, \delta_{E_i}\right)= \int \phi \,\,d\eta\nonumber
\end{eqnarray}
where~$\mathbb{P}_{\star \mathbf{p}(n)}$ is the pushforward of~$\mathbb{P}$ by~$\mathbf{p}(n)$, i.e.,~$\mathbb{P}_{\star \mathbf{p}(n)}\left(A\right)=\mathbb{P}\left(p(n)\in A\right)$ for any~$A\in\mathcal{F}$. Since the above holds for all bounded and continuous~$\phi$ and~$\eta=\sum_{i=1}^M \mathbb{P}\left(\Omega_i\right)\delta_{E_i}\in {\sf co}\left\{\delta_{E_{1}},\ldots,\delta_{E_{M}}\right\}$, the desired convergence in distribution follows.
\end{proof}
Corollary~\ref{th:dual} asserts the existence of a map
\begin{equation}
\begin{array}{ccc}
\mathcal{H}\,:\,\mathcal{P}\left(\left[0,1\right]^{N \times N}\right) & \longrightarrow & {\sf co}\left\{\delta_{E_{1}},\ldots,\delta_{E_{M}}\right\}\\
\mu_{0} & \longmapsto & \mu_{\infty}
\end{array}
\end{equation}
from the set of probability measures on~$\left[0,1\right]^{N \times N}$ to the set of probability measures with support on~$\mathcal{E}$, that maps any initial distribution~$\mu_{0}\in\mathcal{P}\left(\left[0,1\right]^{N \times N}\right)$ to its corresponding attractor~$\mu_{\infty}$. Fully understanding the dynamical system~$\left(\mathbf{p}(n)\right)$ is equivalent to characterizing~$\mathcal{H}$. While we do not obtain an explicit closed form formula for~$\mathcal{H}$, one of our goals in Section~\ref{sec:simula} is to obtain some intuitive characterization of this map.

Figure~\ref{fig:flowchart} shows the flow-chart with the interdependencies among the theorems developed in this subsection.
\begin{figure}
  \centering
  \includegraphics[width=.8\linewidth]{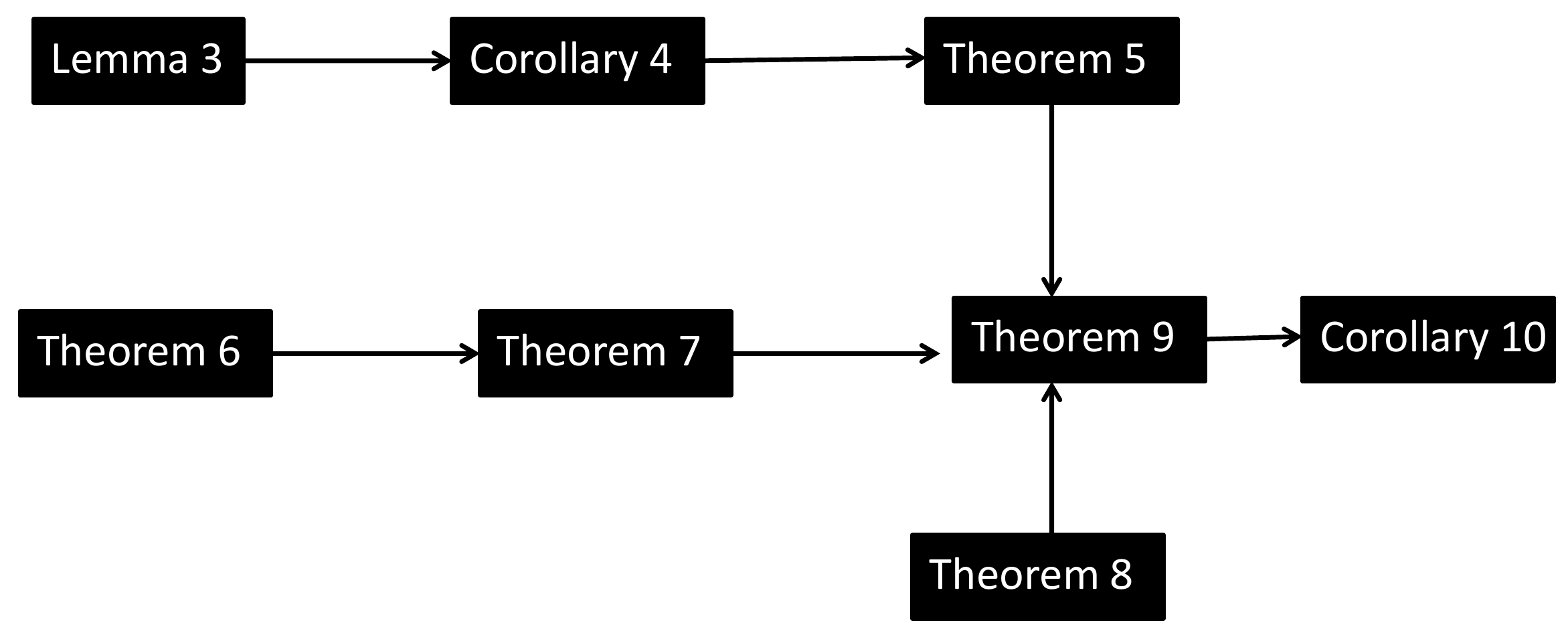}
  \caption{Flow chart of interdependencies among the theorems.}
  \label{fig:flowchart}
\end{figure}


\subsection{Convergence Result Revisited: Calls with Fading Perturbations}\label{subsec:pert}

In this subsection, we assume the same dynamics as in~\eqref{eq:stoc22} for the process~$\left(\mathbf{p}(n)\right)$, with the assumptions described in the first paragraph of this section, except that the likelihood of a call may be subject to perturbations at random times as follows:
\begin{equation}
\mathbb{P}\left(C_{ij}(n)=1\left|\mathbf{p}(n),\,x_{i}(n)\right.\right) = p_{ij}(n)\mathbf{1}_{\left\{x_i(n)=0\right\}}+k_{ij}\left(n,p_{ij}(n)\right)\mathbf{1}_{\left\{x_i(n)=1\right\}},\label{eq:pertlaw}
\end{equation}
where~$x_i(n)$ is the indicator of a perturbation at time~$n$: if~$x_i(n)=1$, then node~$i$ is in perturbation mode and~$i$ calls~$j$ with probability~$k_{ij}\left(n,p_{ij}(n)\right)$; otherwise,~$x_i(n)=0$, i.e.,~$i$ is not in perturbation mode and it follows the standard dynamics based on~$p_{ij}(n)$, and~$\left(k(n,x)\right)_n$ is any (possibly random) sequence of functions on the interval~$\left(0,1\right)$.

Note that, for a general stochastic process~$\left(\mathbf{x}(n)\right):=\left(x_1(n),\ldots,x_{N}(n)\right)$, the dynamical system~$\left(\mathbf{p}(n)\right)$ does not necessarily converge in this perturbed case. If we further assume that the perturbation factor decays over time as follows
\begin{equation}
\sum_{n}\mathbb{P}\left(x_{i}(n)=1\right) < \infty, \label{eq:BC}
\end{equation}
then we obtain convergence as we show in the next theorem.
\begin{theorem}\label{th:finalwexplor}
Assume that~$\left(\mathbf{p}(n)\right)$ follows the dynamics in~\eqref{eq:stoc22} with the calling law given in~\eqref{eq:pertlaw} and the vector random sequence~$\left(\mathbf{x}(n)\right)$ satisfying condition~\eqref{eq:BC}. Let the assumption on~$\left(R(n)\right)$ to be the same as in Theorem~\ref{th:final}. Then,
\begin{equation}
\mathbb{P}\left(\mathbf{p}(n)\overset{n\to \infty}\longrightarrow E_{\ell}\in\mathcal{E},\,\mbox{for some}\,\ell\right)=1.\nonumber
\end{equation}
\end{theorem}

\begin{proof}
Define the sequence of events (indexed by~$N$)
\begin{equation}
Ev_{N}\overset{\Delta}=\left\{\omega\in \Omega\,:\, x(n,\omega)=0,\,\forall{n\geq N}\right\},\nonumber
\end{equation}
and observe that~$Ev_N \subset Ev_{N+1}$. Now, note that, when the stochastic dynamical system~$\left(\mathbf{p}(n)\right)$ is restricted to the event~$Ev_N$ (for some~$N$), it follows eventually the standard dynamics (without perturbation), and, therefore, by Theorem~\ref{th:final}, $\left(\mathbf{p}(n)\right)$ converges almost surely to a matrix in~$\mathcal{E}$ on the event~$Ev_{N}$, for any~$N$. More formally,
\begin{equation}
\mathbb{P}\left(\mathbf{p}(n)\overset{n\rightarrow \infty}\longrightarrow E_{\ell},\,\mbox{for some }\ell\, \left| \,Ev_{N}\right.\right)=1,\nonumber
\end{equation}
i.e.,~$\left(\mathbf{p}(n)\right)$ converges almost surely in the induced probability subspace
\begin{equation}
\left(\widetilde{\Omega},\widetilde{F},\widetilde{\mathbb{P}}\right):=\left(Ev_N,\,\mathcal{F}\left|_{Ev_N},\,\mathbb{P}\left(\cdot\left|\right. Ev_N\right)\right.\right).\nonumber
\end{equation}
Let
\begin{equation}
A\overset{\Delta}=\left\{\omega\in\Omega\,:\,\mathbf{p}(n,\omega)\overset{n\rightarrow \infty}\longrightarrow E_{\ell},\,\mbox{for some }\ell\right\},\nonumber
\end{equation}
then,
\begin{equation}
\mathbb{P}\left(A\right)=\underbrace{\mathbb{P}\left(A\left|\right. Ev_N\right)}_{=1}\mathbb{P}\left(Ev_N\right)+\mathbb{P}\left(A\left|\right. \overline{Ev_N}\right)\mathbb{P}\left(\overline{Ev_N}\right)\geq \mathbb{P}\left(Ev_N\right),\nonumber
\end{equation}
for all~$N\in\mathbb{N}$. This implies in particular that
\begin{equation}\label{eq:lasteqconv}
\mathbb{P}\left(\mathbf{p}(n)\overset{n\rightarrow \infty}\longrightarrow E_{\ell},\,\mbox{for some }\ell\right)\geq \mathbb{P}\left(Ev_{N}\right)\nearrow \mathbb{P}\left(x(n)=0,\mbox{eventually} \right)=1,
\end{equation}
where the last equality in~\eqref{eq:lasteqconv} follows from the assumption in~\eqref{eq:BC} and the Borel-Cantelli lemma~\cite{williamsprobability}. Thus,~$\mathbf{p}(n)$ converges almost surely to a binary matrix in~$\mathcal{E}$ and the theorem is proved.
\end{proof}
Even though almost sure convergence is preserved under sparse, in time, perturbations, the strong-attractors of the dynamical system~$\left(\mathbf{p}(n)\right)$ are sensitive to perturbations. This may lead to important implications on social networks whose evolution follows the reinforcement and penalization dynamics with fading perturbations as described in this section.

In this section, we proved that under quite broad conditions on the distributions of the response, and the rules of reinforcement, penalization, and fading, any dynamical system in the family in~\eqref{eq:stoc22} converges almost surely to the subset~$\mathcal{E}$ of the set of binary matrices as time goes to infinity, regardless of the initial condition. In other words, for any initial distribution~$\mu_{0}$ on~$\left[0,1\right]^{N \times N}$,~$\left(\mathbf{p}(n)\right)$ converges weakly to a random variable~$\mathbf{p}(\infty)$ with support on~$\mathcal{E}$. The next section illustrates finer properties of the long-term behavior of the system through a set of numerical simulations. We observe macroscopic properties of the attractors by computing the histogram of the number of nodes \emph{versus} out-flow degree of the limiting matrices~$\mathbf{p}(\infty)$ and through other statistics of the emergent networks. We illustrate, in particular, that our dynamical system exhibits the emergence and downfall of latent leaders in the long run.

\section{Numerical Simulations}\label{sec:simula}

The goal of this section is to illustrate via numerical simulations that our family of dynamical systems~\eqref{eq:stoc22} displays a rich set of long-term behaviors that depend on the specific configuration of the parameters involved, namely, the initial distribution of~$\left(\mathbf{p}(n)\right)$, the response distribution, the number of nodes, and the laws RPF (though the sensitivity of the system to the latter is not explored in the simulations). In particular, we illustrate that our family readily captures the emergence of latent leaders that contrast with a crowd of isolated nodes (i.e., nodes with no out-flow connection); this is similar to the behavior observed in multi-organizational networks (EMONs) such as support networks which emerged after the Katrina Hurricane disaster~\cite{katrina}. Also, we apply our family of stochastic dynamical systems~\eqref{eq:stoc22} to model the downfall of veteran leaders due to the rise of new ones, as often observed in social networks (e.g.,~\cite{yingdalu}). The downfall of older dominant leaders is due to competition induced by the limited capacity of response of each node. More formally, we refer to a leader as a node that is ranked amongst the top nodes by their out-flow degree. In the long run, such a leader can reach out to many nodes in the network with guaranteed response. In other words, such node plays a major role in the global decision making and the overall network dynamics.


We assume throughout (unless stated otherwise):

$\bullet$ \textbf{Initial Condition.} $p_{ij}(0)\sim \mbox{Uniform}(0,1)$, i.e., the state of each edge is initially drawn uniformly randomly from the interval~$(0,1)$.

$\bullet$ \textbf{Size of the Network.} $N=200$ nodes;

$\bullet$ \textbf{Laws of Reinforcement and Penalization.}
\begin{equation}\begin{array}{cc}
f_{i,+}(x) = \frac{1+x}{2}, & f_{i,-}(x) = g_{i,-}(x) = \frac{x}{2},
\end{array}\label{eq:rfprules}
\end{equation}
and note that~$f_{i,+}\in \mathcal{C}_{+}\cap \mathcal{P}_{+}$, and~$f_{i,-},g_{i,-}\in \mathcal{C}_{-}\cap \mathcal{P}_{-}$;

$\bullet$ \textbf{Calls.} A call from node~$i$ to node~$j$ is made with probability~$p_{ij}(n)$ at time~$n$ for simulations in Subsection~\ref{subsec:latent}, and it will be subject to fading perturbations in the simulations in Subsection~\ref{subsec:rise};

$\bullet$ \textbf{Responses.} $Q_{i}=10$ is the capacity of response of each node~$i$. In words: i) if any node~$i$ receives at most~$10$ incoming calls, it must answer all of them; ii) if any node receives more than~$Q_{i}=10$ calls, then it must choose~$10$ of them to answer. The rule to choose which one to answer is based on network level aspects of the callers as explained next. Recall that our convergence results, Theorems~\ref{th:final}-\ref{th:finalwexplor} and Corollary~\ref{th:dual}, hold true regardless of the bias rule (and indeed, in all simulations, the dynamical system~$\left(\mathbf{p}(n)\right)$ was empirically observed to converge).

Regarding the response type, it can be of two types: \emph{democratic} or \emph{biased}. In the democratic mode, we assume that, if a node~$i$ receives more calls than its capacity~$Q_{i}$, then it chooses uniformly randomly~$Q_{i}$ nodes among the callers to answer -- which, as a result, will reinforce the corresponding connections. In the biased case, we assume that, once~$i$ receives more calls than its capacity, it chooses the callers to respond to based on their out-flow degree, i.e.,~$\sum_{\ell}p_{k\ell}$. In particular, node~$i$ will select sequentially one by one the callers to respond. It chooses a caller~$k_1$, first, with probability:
\begin{equation}
\frac{\exp\left(\theta \sum_{\ell}p_{k_1\ell}\right)}{\sum_{m \mbox{ called }i}\exp\left(\theta \sum_{\ell}p_{m\ell}\right)}.\nonumber
\end{equation}
Then, it removes this selected caller~$k_1$ from the list and selects the next caller by drawing from
\begin{equation}
\frac{\exp\left(\theta \sum_{\ell}p_{k_2\ell}\right)}{\sum_{m\neq k_1 \mbox{ called }j}\exp\left(\theta \sum_{\ell}p_{m\ell}\right)},\nonumber
\end{equation}
and so on -- until it exhausts its capacity. The parameter~$\theta\in\left(0,\infty\right)$ accounts for the bias in the choice. In particular, if~$\theta$ is large, it means that the responder will strongly weight reputable nodes -- where the reputation is drawn from the out-flow degree. If~$\theta=0$, then the system operates in a fully democratic mode.

\textbf{Remark on convergence.} Note that the convergence in Theorems~\ref{th:final},~\ref{th:dual} and~\ref{th:finalwexplor}, does not rely on the particular response law to select the callers, but only on the fact that each node~$i$ responds necessarily to
\begin{equation}
\min\left\{C_i,\,\#\mbox{ of callers to }i\right\}\nonumber
\end{equation}
at each time~$n$, or in other words, on the fact that~$\left(\mathbf{R}(n)\right)$ obeys~\eqref{eq:responsecapa}. Also, the particular choice of the RPF rules, in~\eqref{eq:rfprules}, is not relevant for convergence as long as they belong accordingly to the sets~$\mathcal{P}_{+}\cap \mathcal{C}_{+}$ and~$\mathcal{P}_{-}\cap \mathcal{C}_{-}$. Indeed, we observed convergence in all simulations presented in this section. As mentioned before, even though the~$0-1$~law for the dynamical system~$\left(\mathbf{p}(n)\right)$ is robust to the initial conditions, the different RPF laws and the type of response, their particular choices certainly affect the long term properties of the system, i.e., the attractors, as it will be illustrated in Subsection~\ref{subsec:latent}.



\subsection{Emergence of Latent Leaders}~\label{subsec:latent}

In Figs.~\ref{fig:test2} and~\ref{fig:test1}, we assume a democratic and biased dynamical system, respectively. For the simulations in both figures, we assume the same initial conditions (in distribution), and observe quite distinct typical attractors. In particular, we observe that as one allows for bias -- for this case,~$\theta=2$ and the number of nodes~$N=200$ -- the limiting behavior of the system is characterized typically by the emergence of a few nodes with large out-degree and close to~$90\%$ of nodes with few or no out-flow connections.

\begin{figure}
\centering
\begin{subfigure}{.5\textwidth}
  \centering
  \includegraphics[width=1\linewidth]{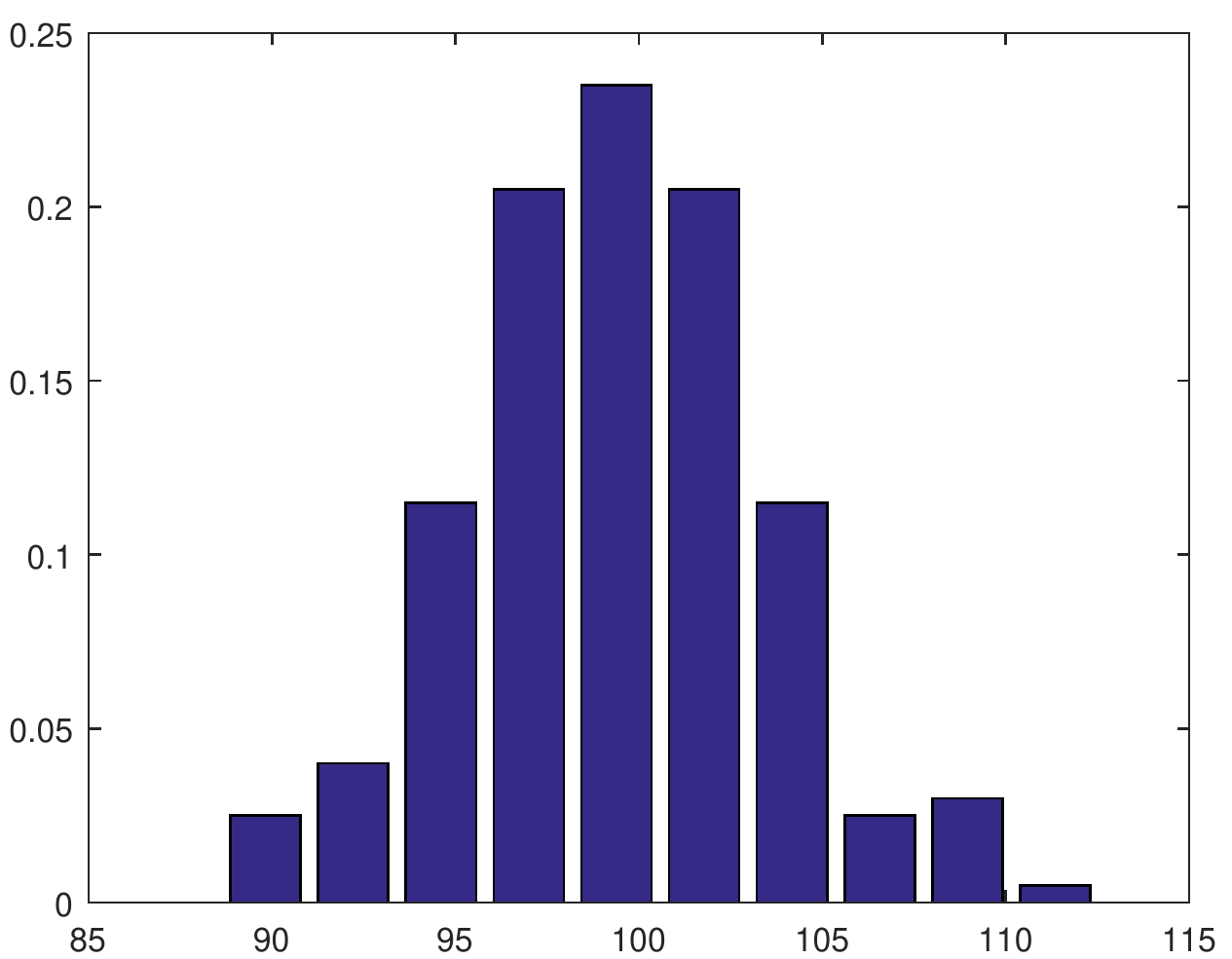}
  \caption{Initial histogram}
  \label{fig:sub1}
\end{subfigure}%
\begin{subfigure}{.5\textwidth}
  \centering
  \includegraphics[width=1\linewidth]{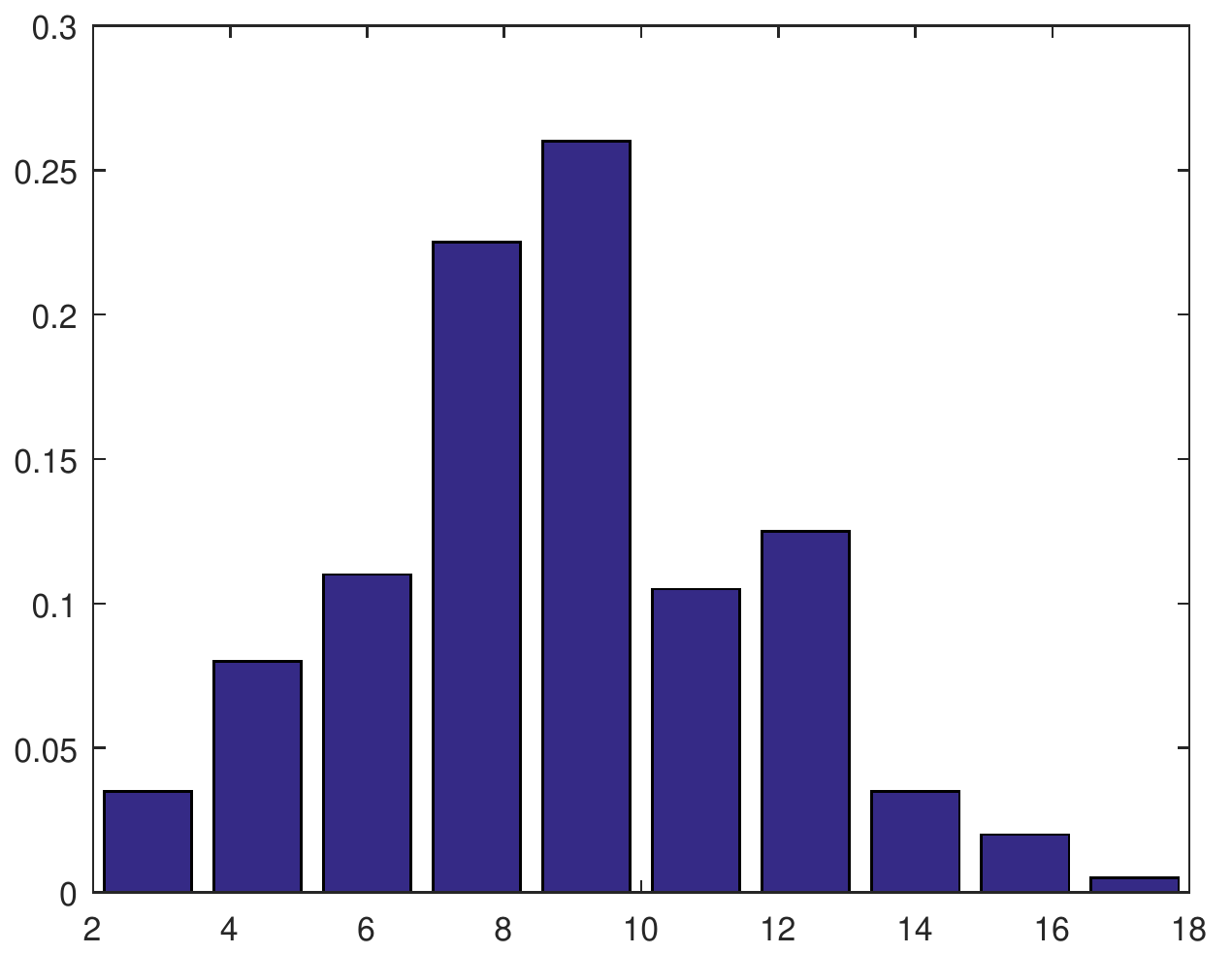}
  \caption{Final Histogram}
  \label{fig:sub2}
\end{subfigure}
\caption{Democratic system with~$\theta=0$.}
\label{fig:test2}
\end{figure}

\begin{figure}
\centering
\begin{subfigure}{.5\textwidth}
  \centering
  \includegraphics[width=1\linewidth]{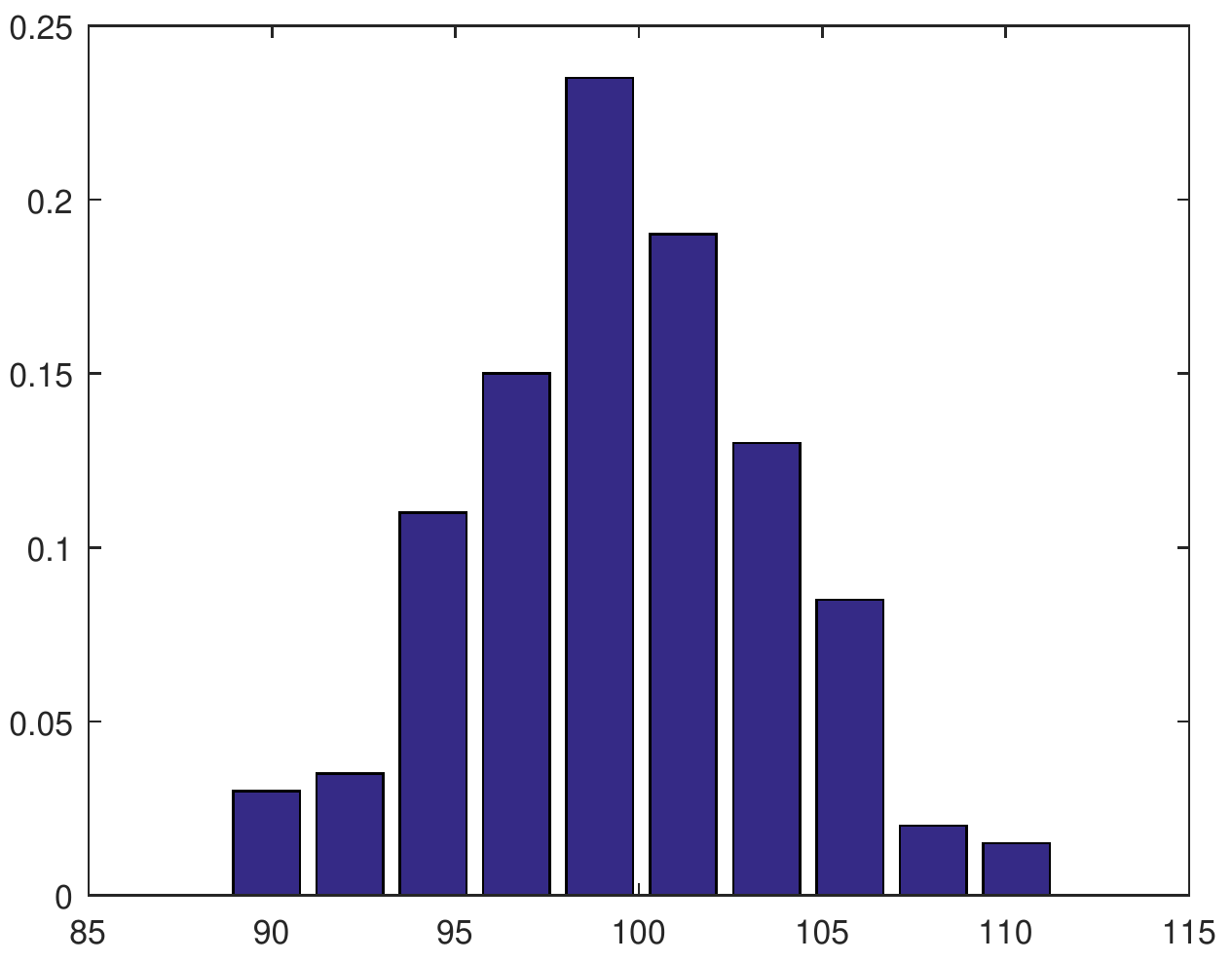}
  \caption{Initial histogram}
  \label{fig:sub1}
\end{subfigure}%
\begin{subfigure}{.5\textwidth}
  \centering
  \includegraphics[width=1\linewidth]{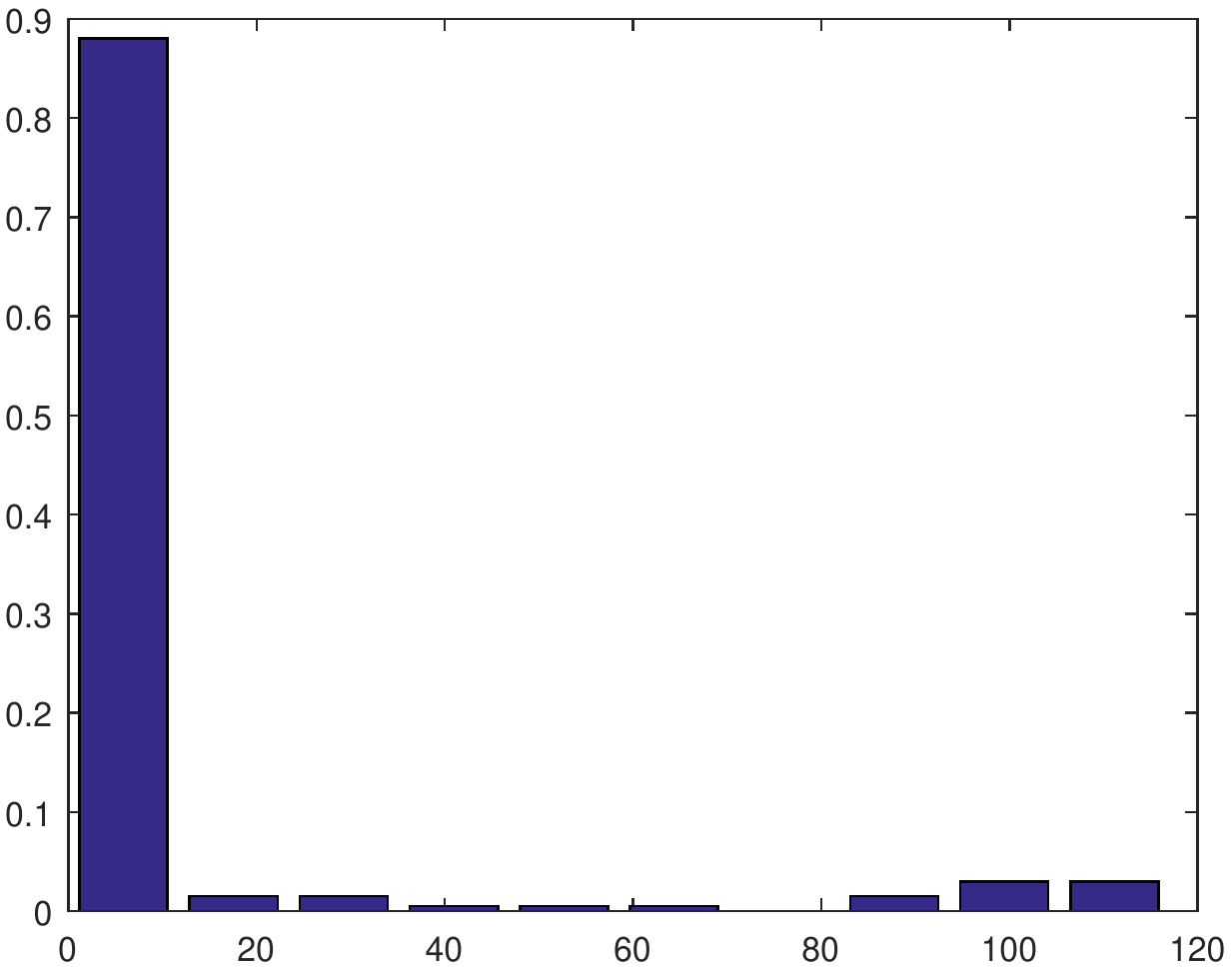}
  \caption{Final Histogram}
  \label{fig:sub2}
\end{subfigure}
\caption{Biased system with~$\theta=2$.}
\label{fig:test1}
\end{figure}

The greater the number of nodes, the steeper tends to be the power-law like behavior of the limiting typical graph, i.e., the histogram shape is not scale-invariant. Fig.~\ref{fig:test33} illustrates the limiting graph for different number of nodes, when the bias is small,~$\theta=0.1$. The response capacity~$Q_{i}$ for each node~$i$ is normalized to be~$5\%$ of the total number of nodes in the network.

\begin{figure}
\centering
\begin{subfigure}{.4\textwidth}
  \centering
  \includegraphics[width=1.0\linewidth]{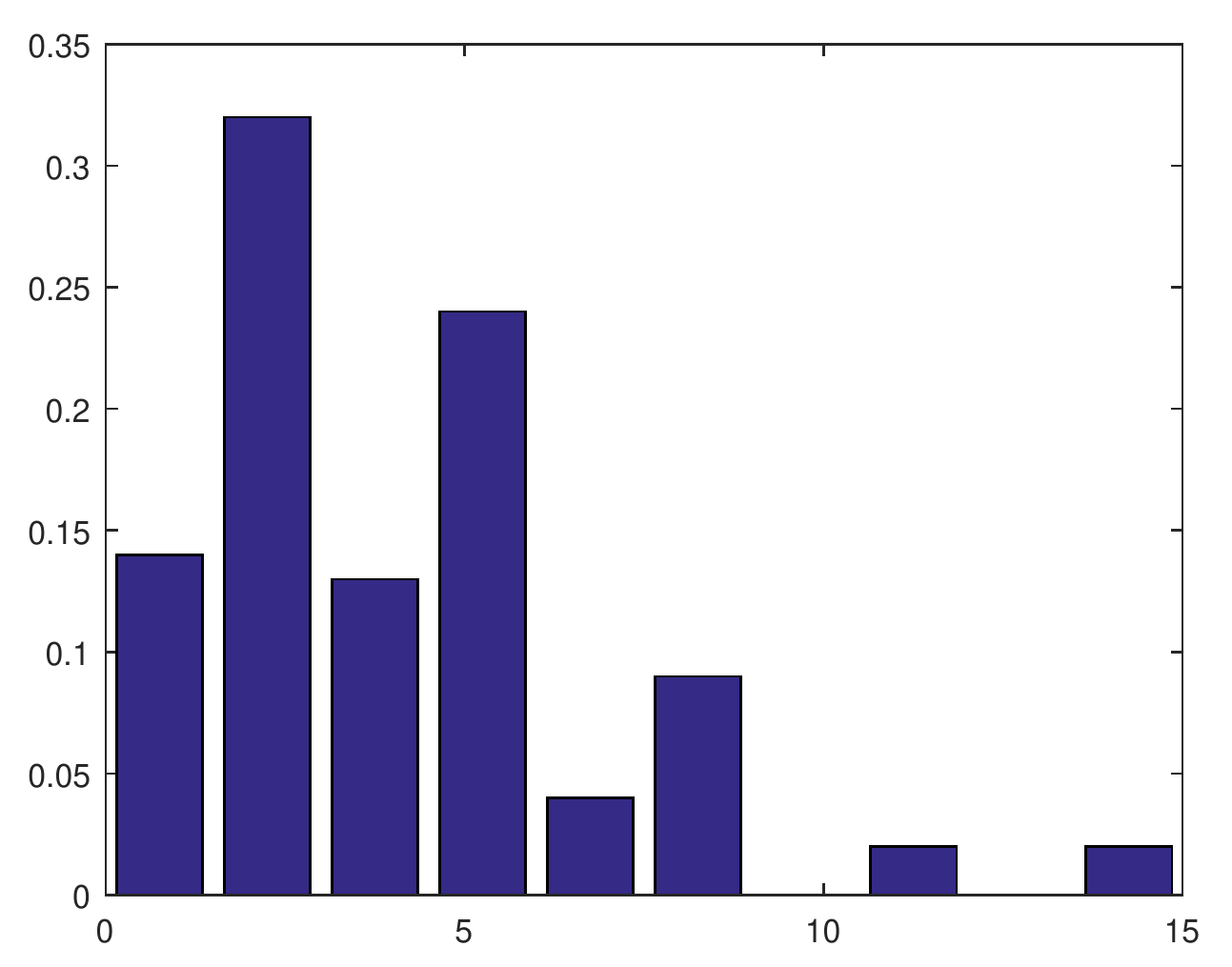}
  \caption{$100$ nodes with~$\theta=0.1$ and capacity of response~$5$.}
  \label{fig:sub1}
\end{subfigure}%
\begin{subfigure}{.4\textwidth}
  \centering
  \includegraphics[width=1.0\linewidth]{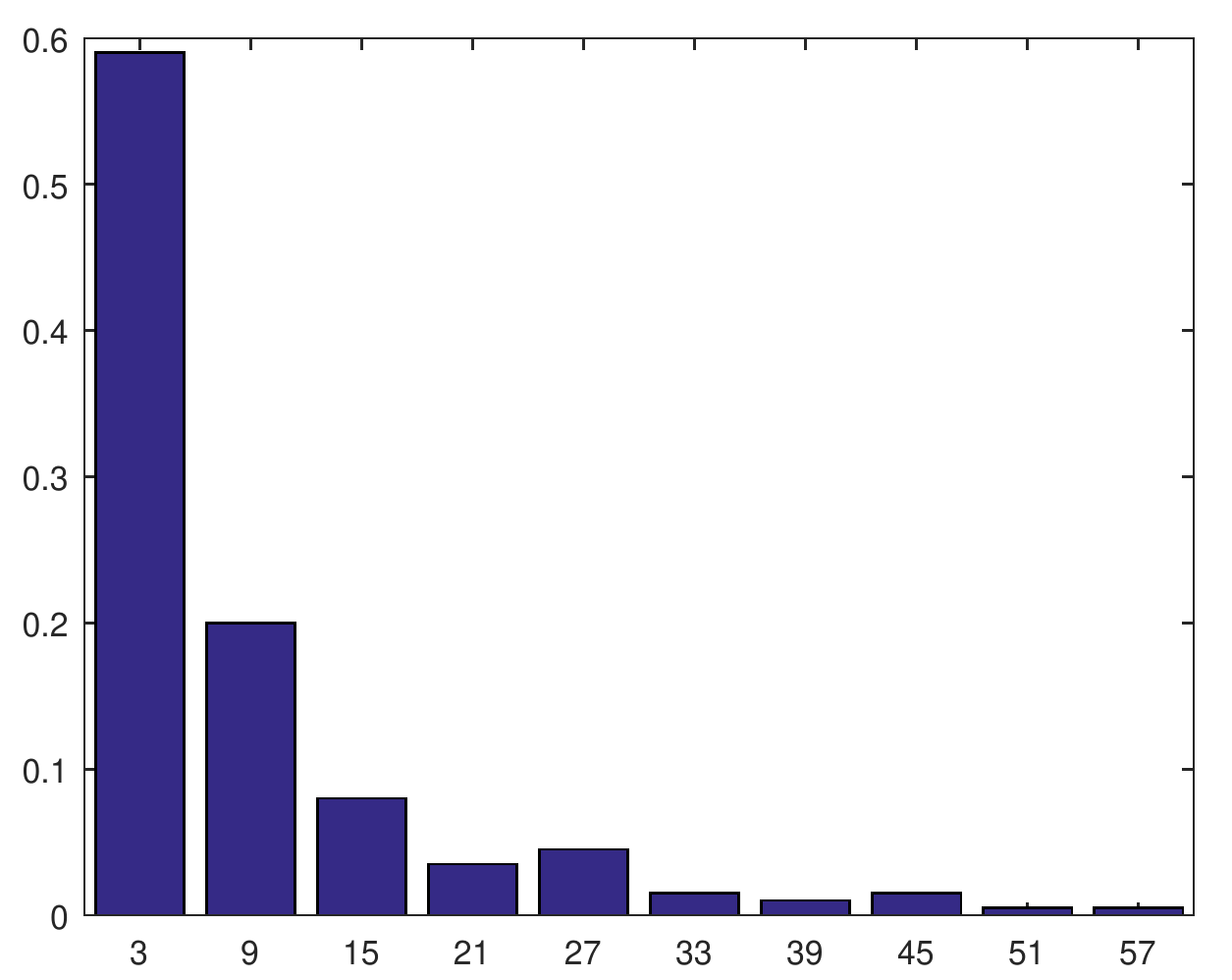}
  \caption{$200$ nodes with~$\theta=0.1$ and capacity of response~$10$.}
  \label{fig:sub1}
\end{subfigure}
\begin{subfigure}{.4\textwidth}
  \centering
  \includegraphics[width=1.0\linewidth]{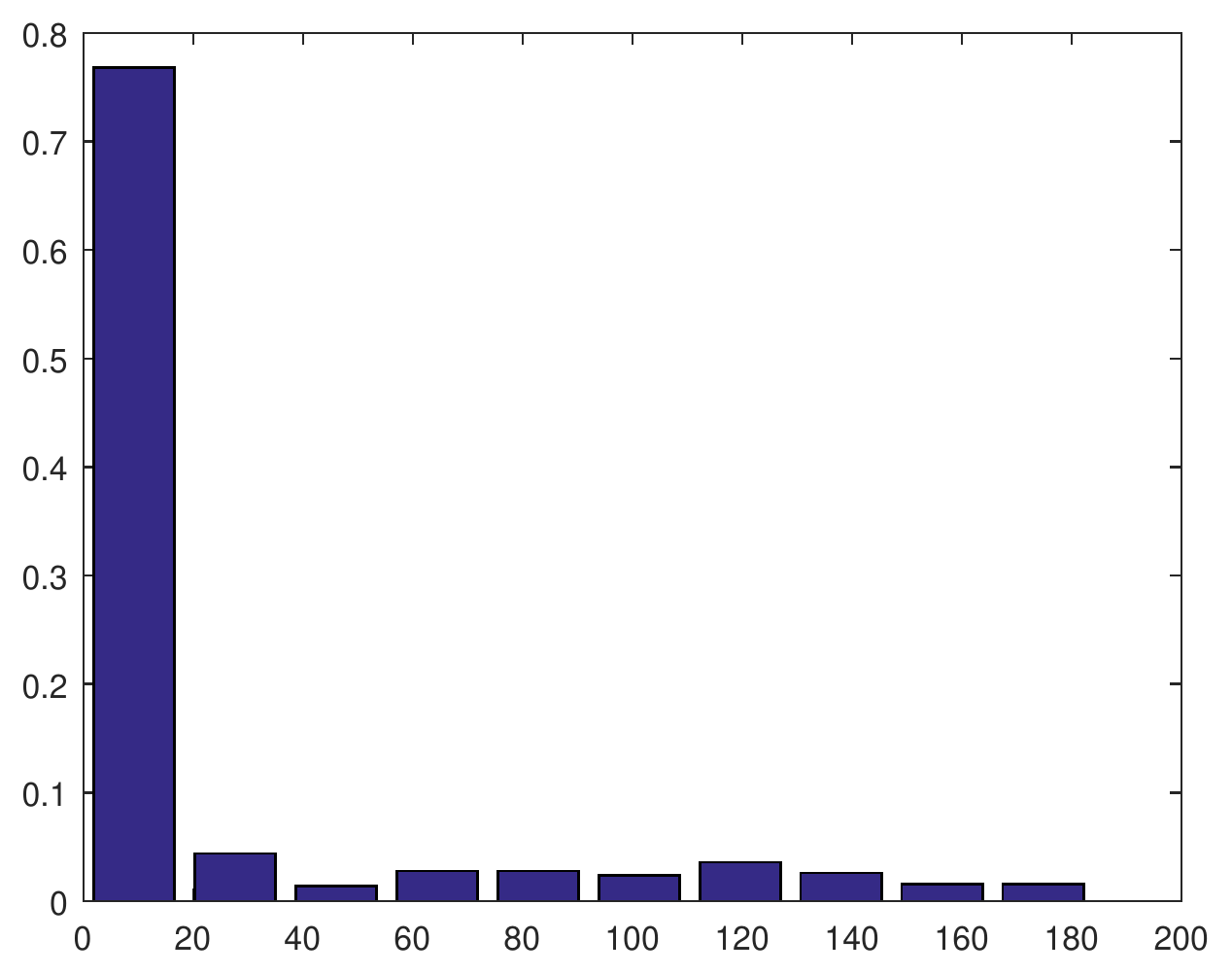}
  \caption{$500$ nodes with~$\theta=0.1$ and capacity of response~$25$.}
  \label{fig:sub2}
\end{subfigure}
\begin{subfigure}{.4\textwidth}
  \centering
  \includegraphics[width=1.0\linewidth]{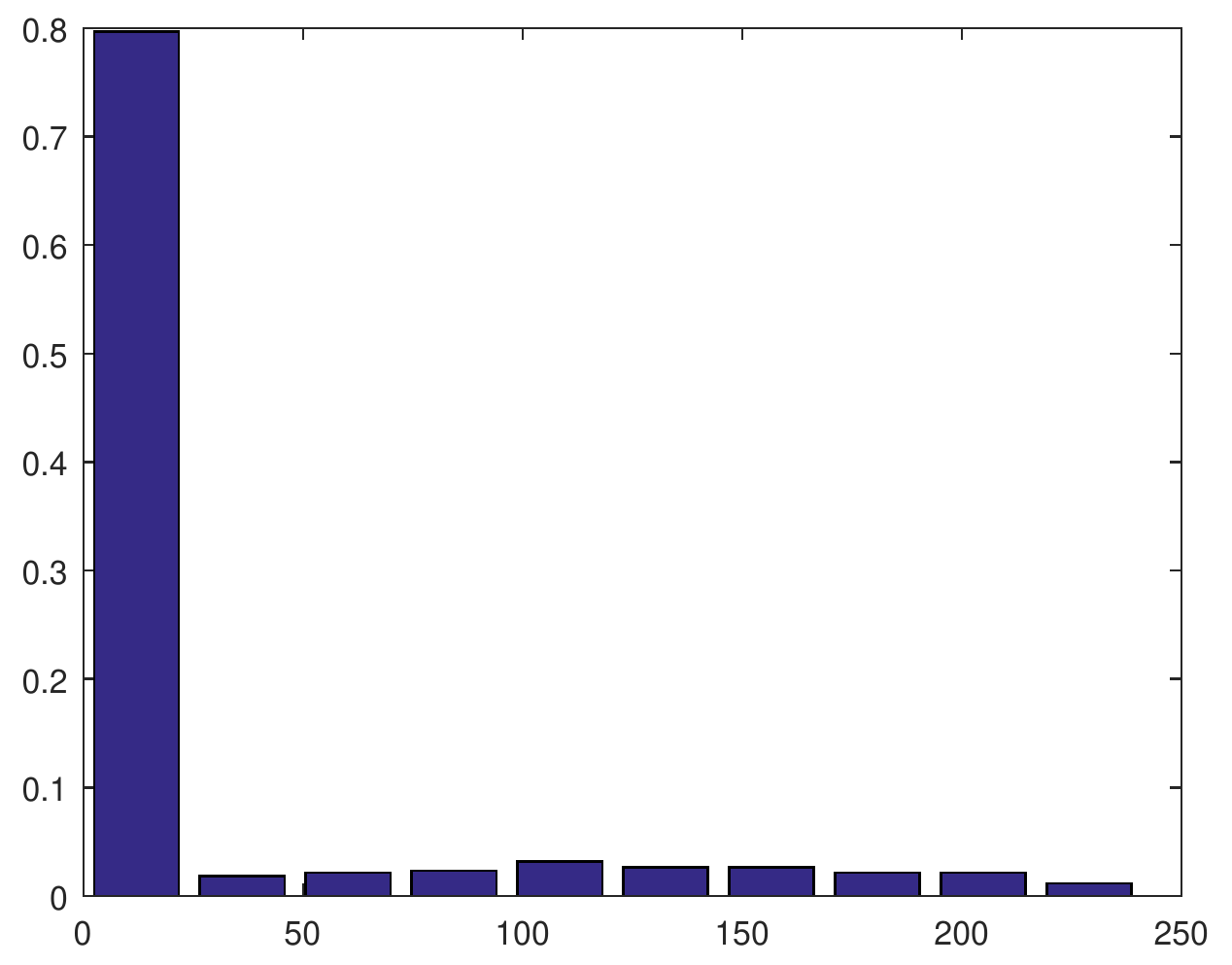}
  \caption{$600$ nodes with~$\theta=0.1$ and capacity of response~$30$.}
  \label{fig:sub2}
\end{subfigure}
\begin{subfigure}{.4\textwidth}
  \centering
  \includegraphics[width=1.0\linewidth]{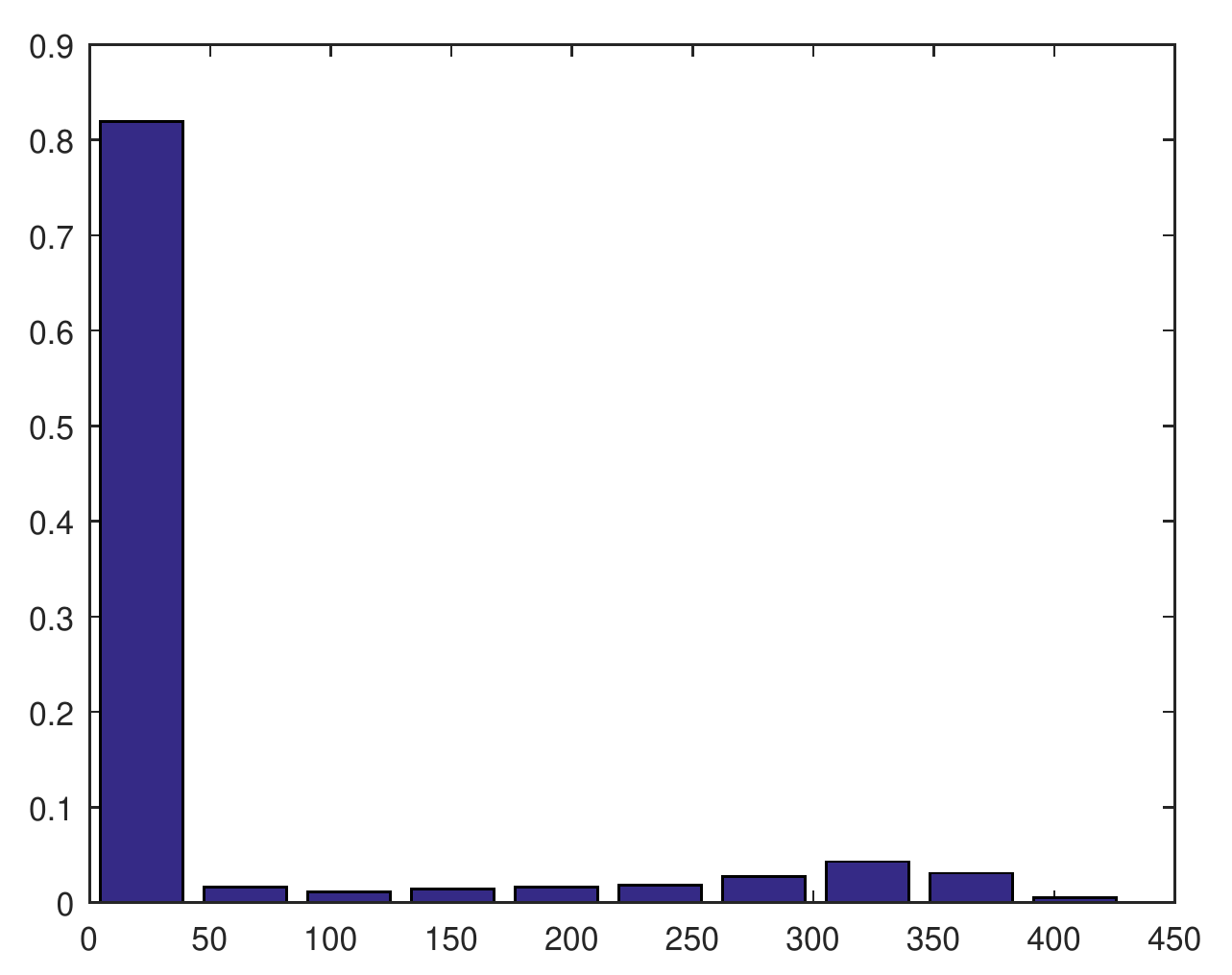}
  \caption{$1000$ nodes with~$\theta=0.1$ and capacity of response~$50$.}
  \label{fig:sub2}
\end{subfigure}
\caption{Dynamical system with bias~$\theta=0.1$ and initial condition drawn from a uniform distribution on the interval $\left(0,1\right)$ at each edge.}
\label{fig:test33}
\end{figure}

The long-term behavior of the system depends not only on the number of nodes, but also on the initial condition (distribution of~$\mathbf{p}(0)$). For instance, if one assumes the consensus initial probability,
\begin{equation}
\mathbf{p}(0)=\frac{1}{N-1}\left(\mathbf{1}\mathbf{1}^{\top}-I_{N\times N}\right),\nonumber
\end{equation}
where~$I_{N\times N}$ is the~$N \times N$ identity matrix and~$\mathbf{1}\in\mathbb{R}^{N}$ is the vector with all entries equal to~$1$, then the emergent leaders are far less dominant as compared to the uniform distribution case as illustrated in Fig.~\ref{fig:test23}\footnote{The histograms presented in this section are obtained using the function `hist' in Matlab. The centering of the bins is automatically provided by this function, and even though the bins are not centered at integers, the process~$\left(\mathbf{p}(n)\right)$ did converge to a binary matrix.}.

\begin{figure}
\centering
\begin{subfigure}{.5\textwidth}
  \centering
  \includegraphics[width=1\linewidth]{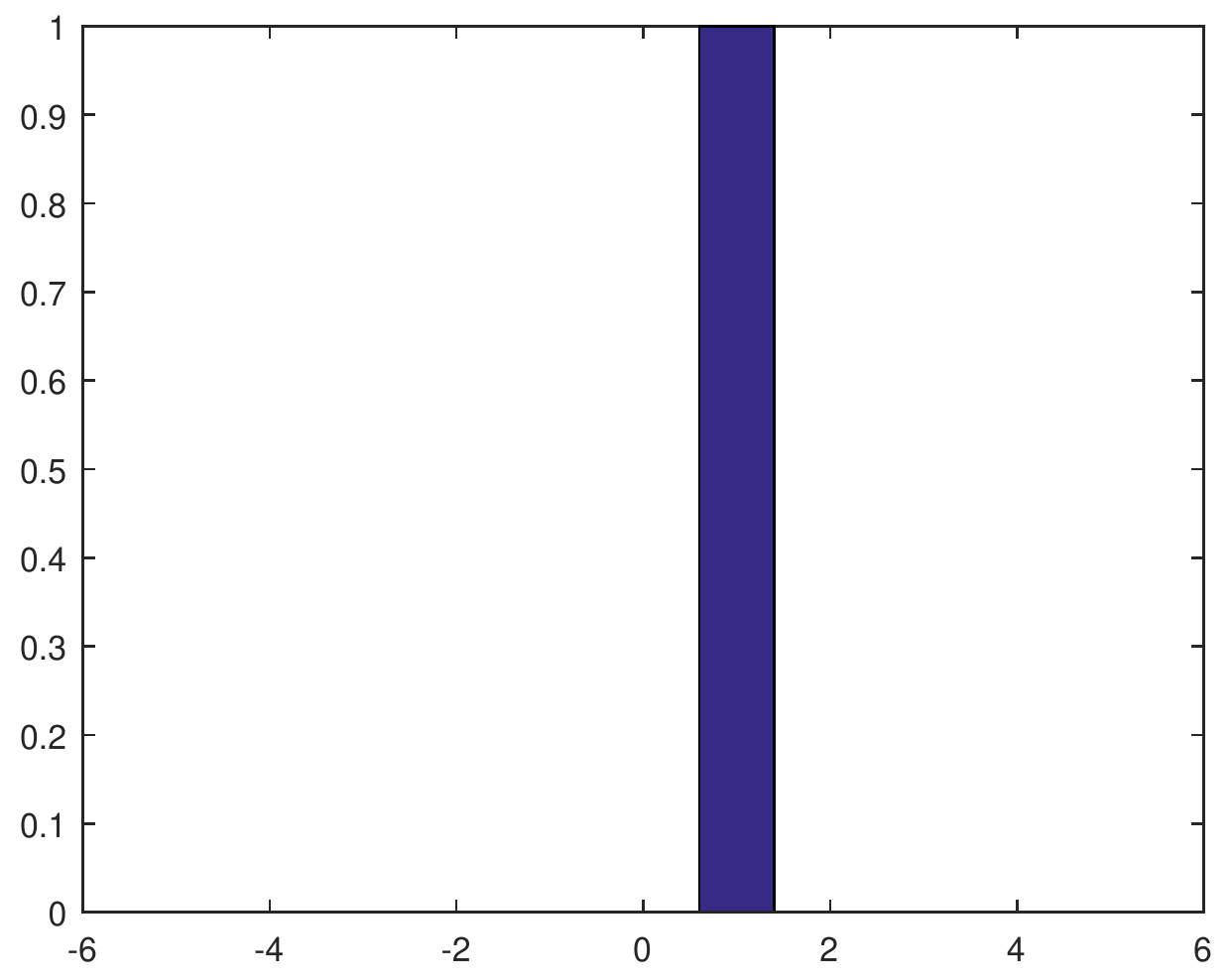}
  \caption{Initial histogram}
  \label{fig:even}
\end{subfigure}%
\begin{subfigure}{.5\textwidth}
  \centering
  \includegraphics[width=1\linewidth]{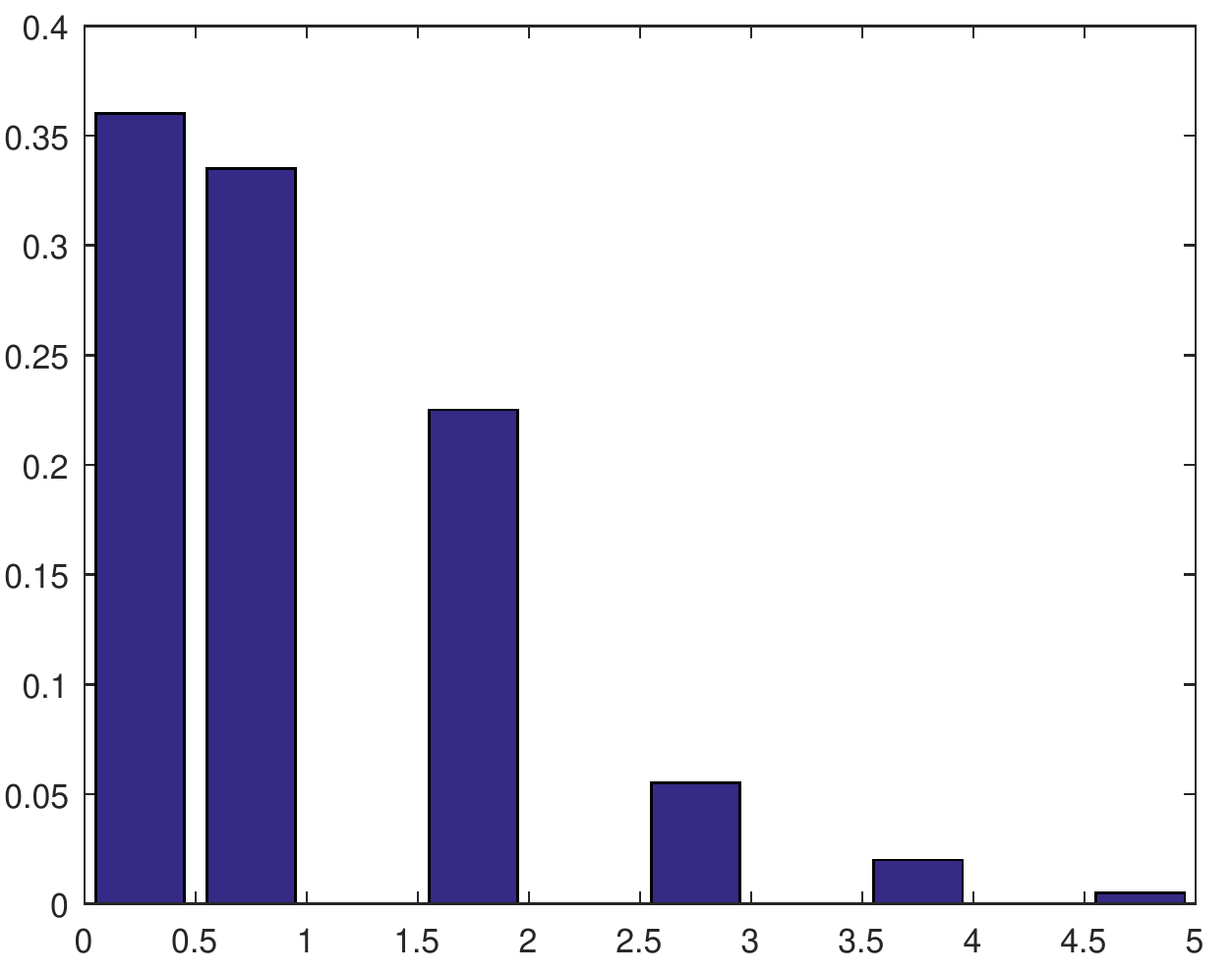}
  \caption{Final Histogram}
  \label{fig:even2}
\end{subfigure}
\caption{System was biased with~$\theta=2$.}
\label{fig:test23}
\end{figure}


One can ask if, for example, the highly dominant nodes in Fig.\ref{fig:test1} were the nodes that were dominant to start with. This is indeed the typical case: quite rarely does one observe a non-dominant node becoming dominant in the system of Fig.\ref{fig:test1} -- uniform initial distribution and bias~$\theta=2$. We remark that the emergent multi-organizational network (EMON) reported in~\cite{katrina} exhibited similar steep qualitative properties: already influential organizations emerge as latent leaders as opposed to a large number of organizations with no (or low) out-flow connections.

In the next subsection, we illustrate that \emph{poor} nodes -- i.e., with low out-flow degree -- may overcome the odds and emerge as leaders in the long run (possibly defying older leaders) if endowed with (fading) exploration as explained below.

\subsection{Downfall of Old and Emergence of New Leaders}\label{subsec:rise}

In this subsection, we present an experiment that illustrates the possibility of a poor node to emerge as a new leader as well as the downfall of a leader due to the competition brought about by new leaders. The main element that leads to the rise of a poor node is exploration. There are plenty of features that may be considered to explain the emergence of leaders, -- e.g., intrinsic quality, experience etc. (e.g.~\cite{yingdalu,krishnan}); in this subsection, we illustrate that our model captures the phenomenon by tracking a few state variables and parameters: calls, responses, capacity of response, and exploration factor.

Exploration is an important factor impacting the evolution of networks. A node explores when it has no connection or only a weak connection to other nodes. Yet, it attempts to call such nodes. The new dynamics are similar to the dynamics in the previous subsection with now
\begin{equation}
\mathbb{P}\left(C_{ij}(n)=1\left|\mathbf{p}(n), \widetilde{x}_i(n)\right.\right)=p_{ij}(n)\mathbf{1}_{\left\{\widetilde{x}_i(n)=0\,\, \vee\,\, p_{ij}(n)\geq \mbox{th}\right\}}+k\mathbf{1}_{\left\{\widetilde{x}_i(n)=1\,\, \wedge \,\, p_{ij}(n)< \mbox{th}\right\}},\label{eq:callsubsec}
\end{equation}
where,~$\widetilde{x}_i(n)$ is the indicator of exploration at time~$n$ of node~$i$: if~$\widetilde{x}_i(n)=1$, then node~$i$ is in exploration mode; otherwise, if~$\widetilde{x}_i(n)=0$, then~$i$ is not in exploration mode. Whenever in exploration mode, if~$p_{ij}(n)$ is \emph{small}, i.e., smaller than a threshold~$\mbox{th}$, then node~$i$ will still contact node~$j$ with likelihood given by~$k$. Therefore, we assume that~$k> \mbox{th}$.

We also assume that the exploration process~$\left(\mathbf{\widetilde{x}}(n)\right)$, decays over time as follows
\begin{equation}
\sum_{n}\mathbb{P}\left(\widetilde{x}_{i}(n)=1\right) < \infty.\label{eq:BCE}
\end{equation}
Note that we are under the purview of Theorem~\ref{th:finalwexplor} in Subsection~\ref{subsec:pert}, as we can define
\begin{equation}
x_i(n):= \widetilde{x}_{i}(n)\mathbf{1}_{\left\{p_{ij}(n)< \mbox{th}\right\}},
\end{equation}
and~$k_{ij}\left(n,p_{ij}(n)\right)=k$ so that~\eqref{eq:callsubsec} can be restated exactly as~\eqref{eq:pertlaw} and moreover
\begin{equation}
\sum_{n}\mathbb{P}\left(x_{i}(n)=1\right) \leq \sum_{n}\mathbb{P}\left(\widetilde{x}_{i}(n)=1\right) < \infty.\nonumber
\end{equation}


We first let the system evolve as described in the previous subsection (i.e., with no exploration) to observe the emergence of leaders and a steep histogram of (number of nodes) \emph{versus} (out-flow degree). The system converges to a binary matrix whose histogram is illustrated in Fig.\ref{fig:even234}. From there, we endow some poor nodes ($9$ of them to be specific) with exploration to overcome the odds and emerge as new leaders. We also show that some of the previous leaders are not able to survive the competition -- remark that the capacity of response of the nodes is bounded by~$10$, which leads to a selective pressure based upon the out-flow degree reputation -- and these nodes fall swiftly.
\begin{figure}
\centering
\begin{subfigure}{.5\textwidth}
  \centering
  \includegraphics[width=1\linewidth]{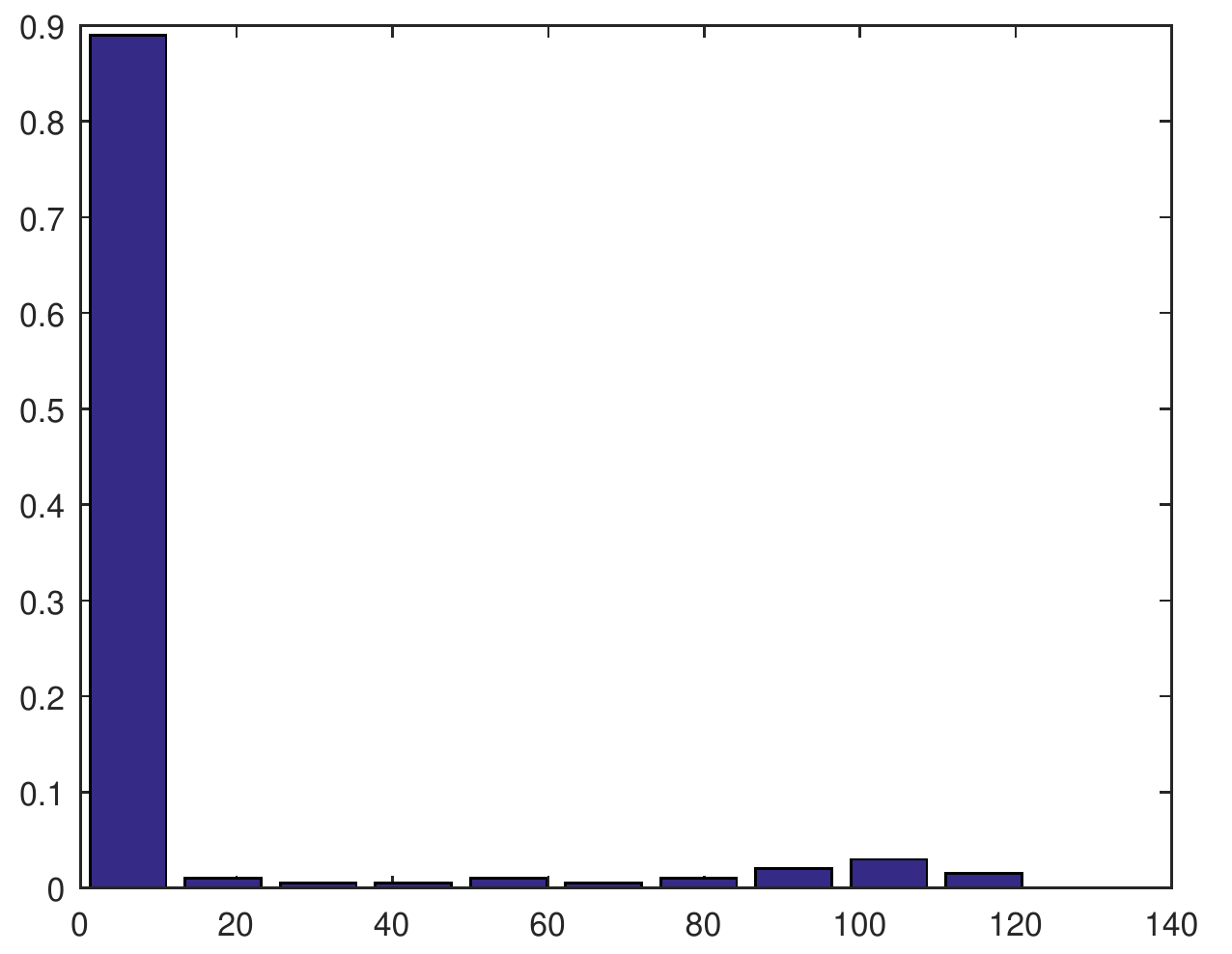}
  \caption{Initial histogram}
  \label{fig:even234}
\end{subfigure}%
\begin{subfigure}{.5\textwidth}
  \centering
  \includegraphics[width=1\linewidth]{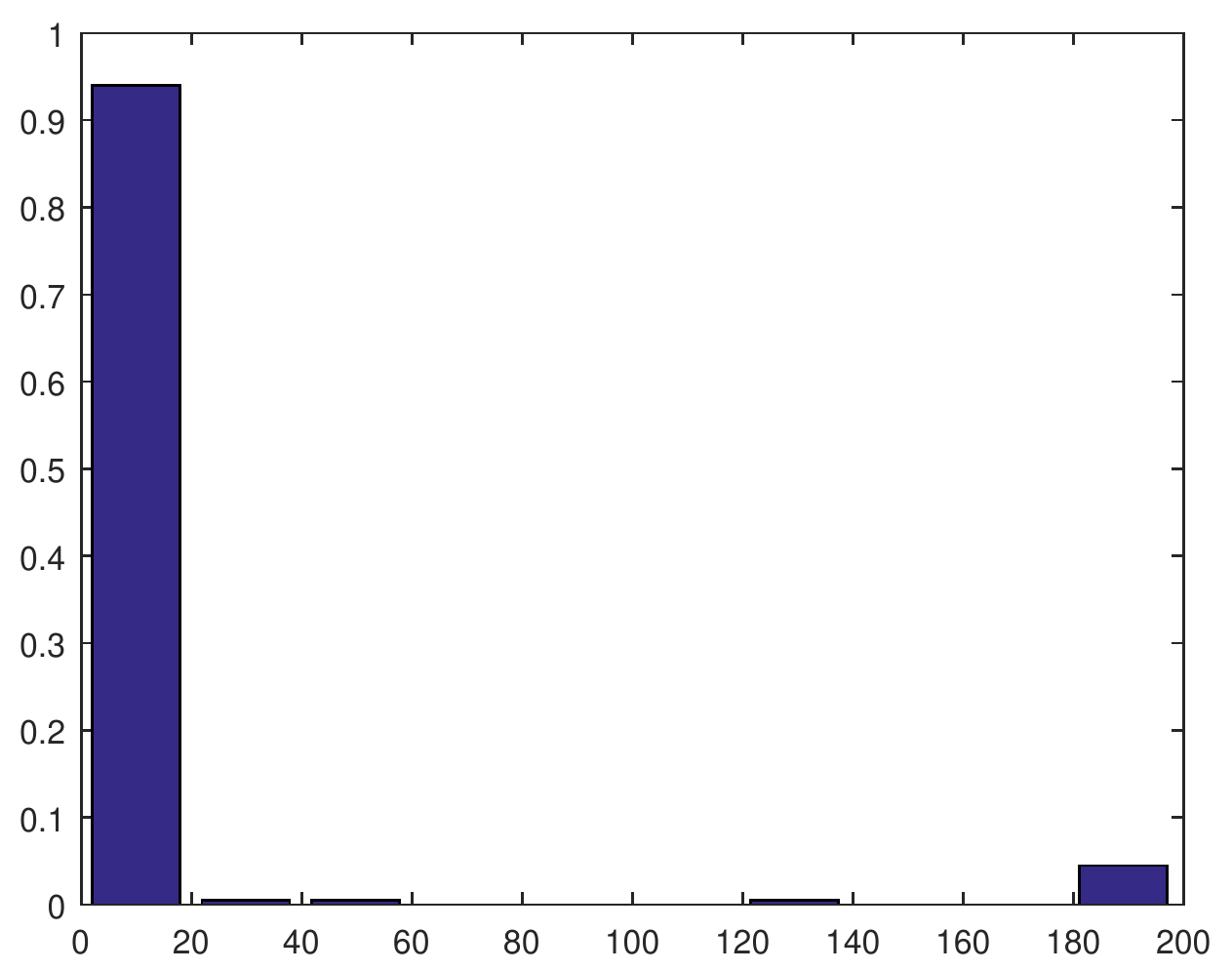}
  \caption{Final Histogram, after the rise of new leaders.}
  \label{fig:even256}
\end{subfigure}
\caption{Biased system with~$\theta=2$.}
\label{fig:alwaysteste2}
\end{figure}
Fig.\ref{fig:even256} depicts the histogram associated with the long term behavior of the system with the new emergent nodes.

Fig.\ref{fig:sub122},\ref{fig:sub12}, and~\ref{fig:sub22} illustrate the evolution of the outflow degree, i.e.,~$\sum_{j}p_{ij}(n)$, of some of the new emergent nodes.
\begin{figure}
\centering
\begin{subfigure}{.4\textwidth}
  \centering
  \includegraphics[width=1.0\linewidth]{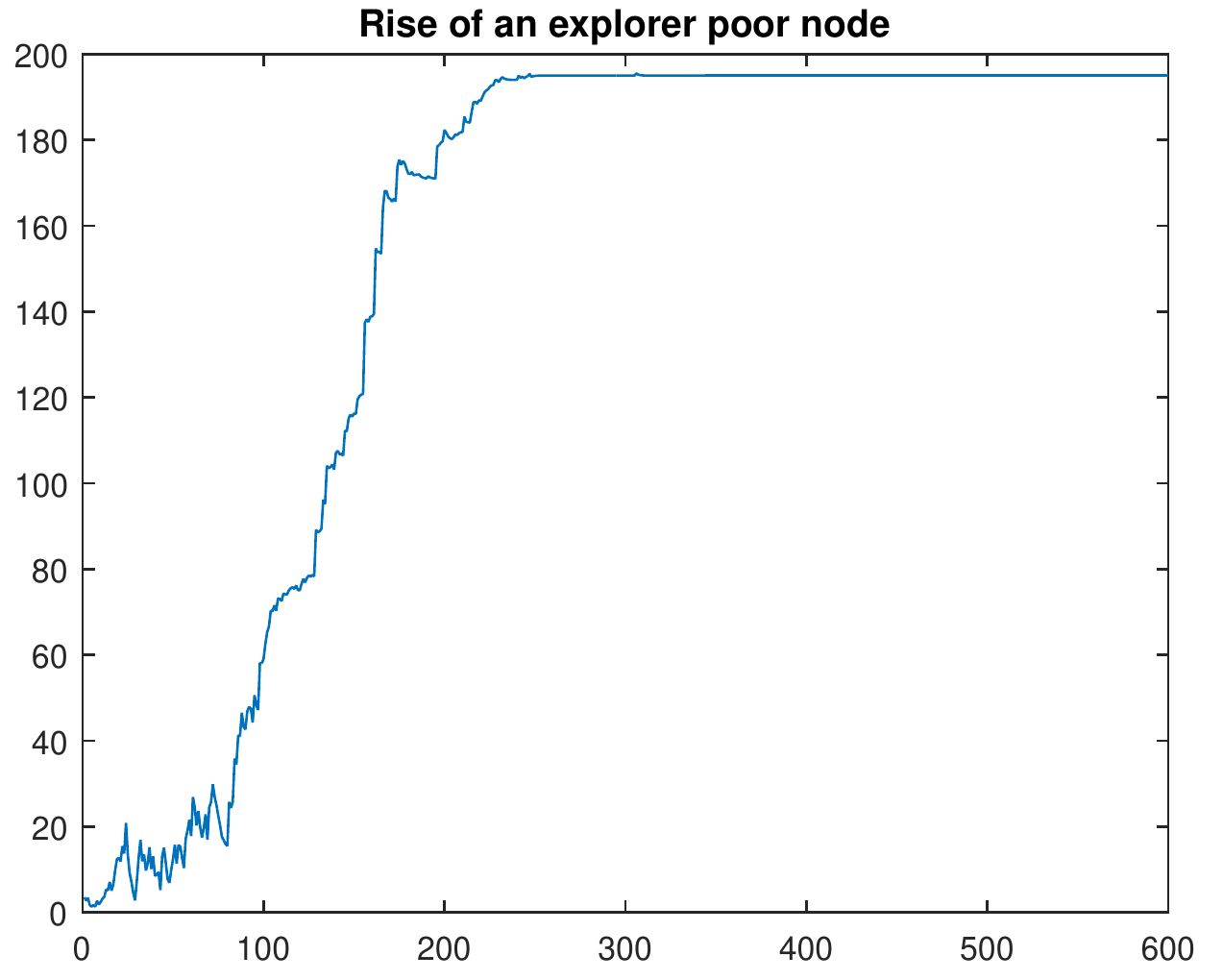}
  \caption{Turbulent rise of a node.}
  \label{fig:sub122}
\end{subfigure}%
\begin{subfigure}{.4\textwidth}
  \centering
  \includegraphics[width=1.0\linewidth]{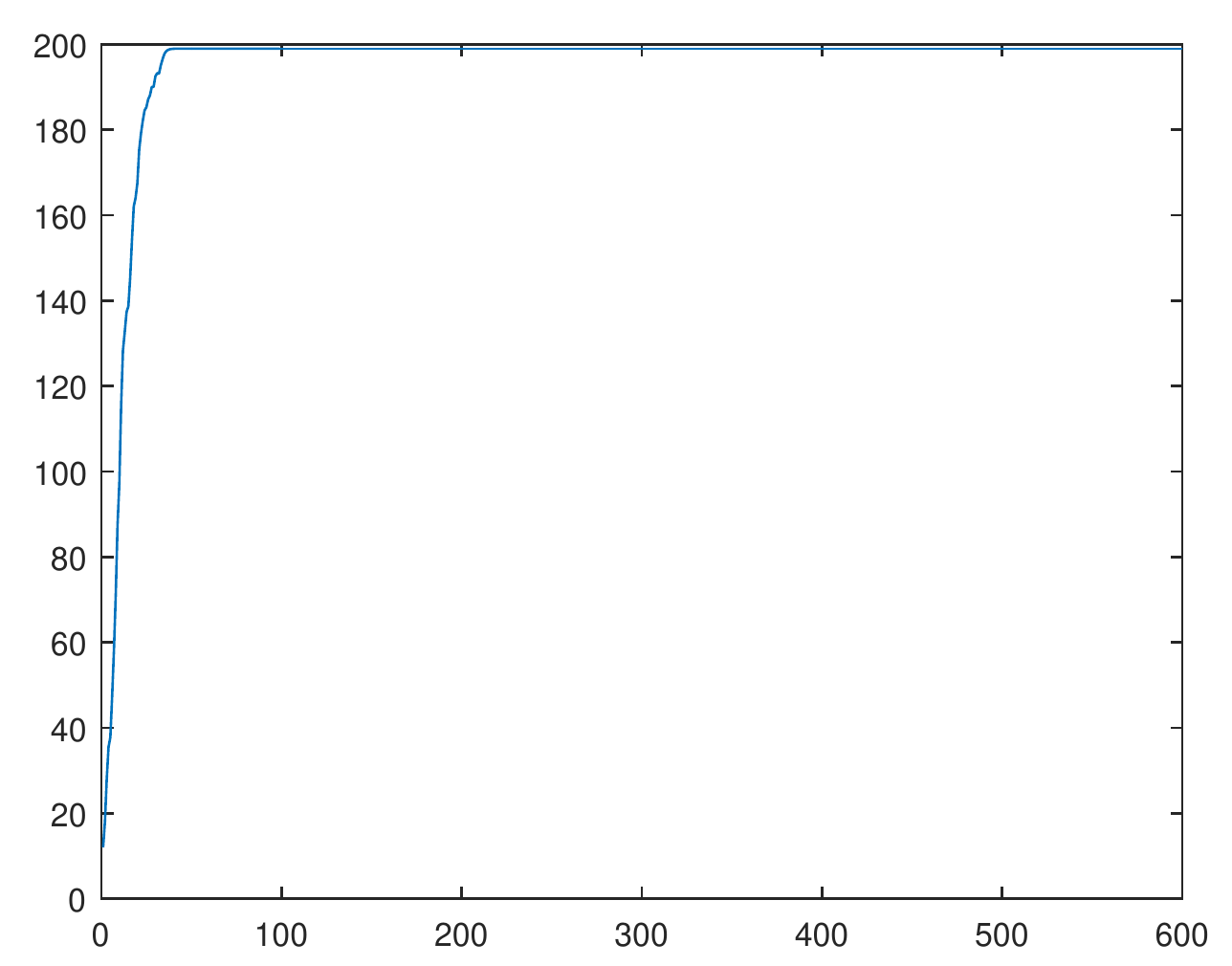}
  \caption{Rapid rise.}
  \label{fig:sub12}
\end{subfigure}
\begin{subfigure}{.4\textwidth}
  \centering
  \includegraphics[width=1.0\linewidth]{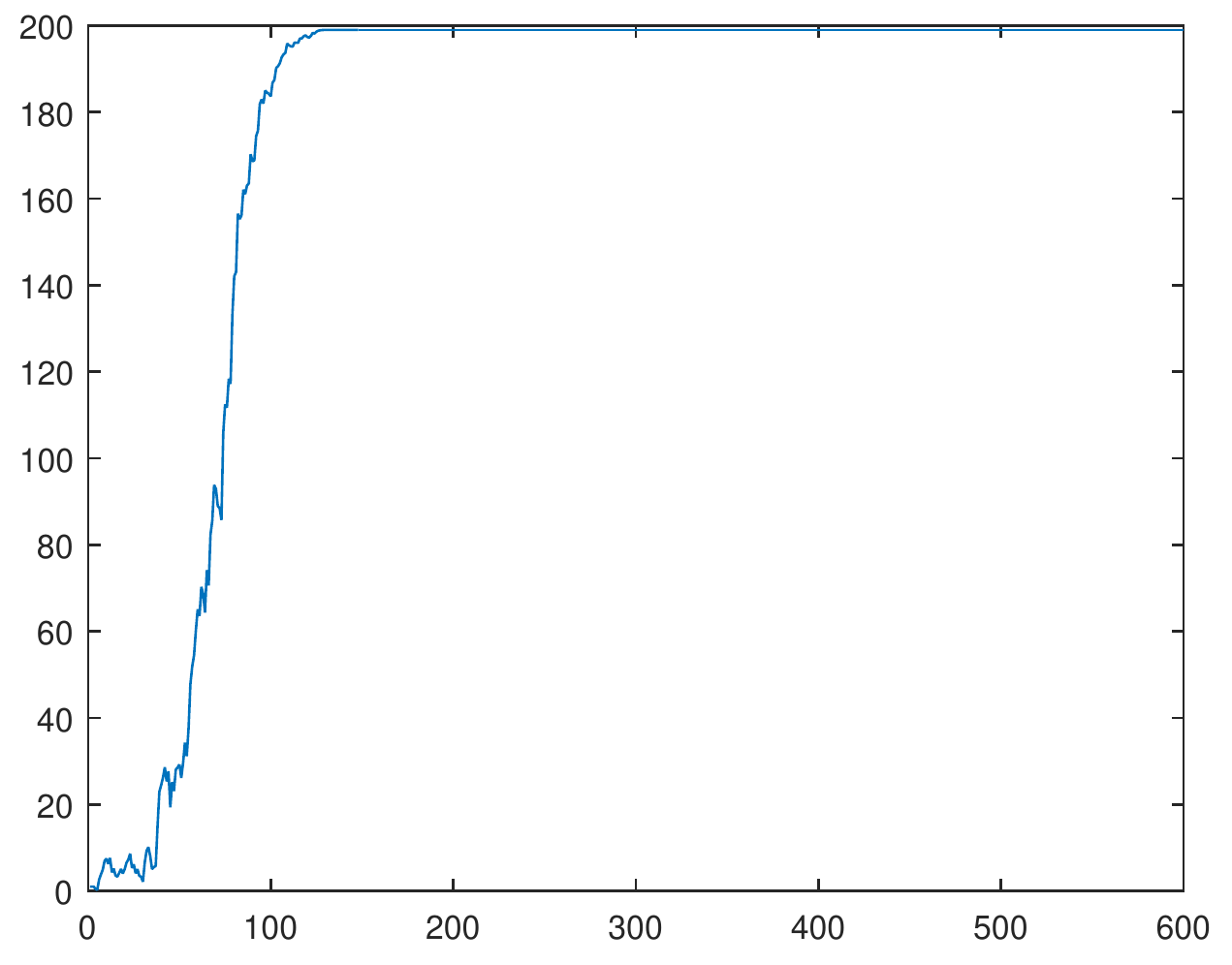}
  \caption{Turbulent rise.}
  \label{fig:sub22}
\end{subfigure}
\end{figure}
Although some nodes may rise swiftly, others may experience initial turbulence -- others may not rise to the top at all. The typical path of a poor node to the top may be depicted as follows: first, it increases its outflow degree by calling nodes that happen to have some sparse response capacity (this is done by random exploration) and then, after gaining some reputation, the node becomes more competitive and more likely to overcome other nodes. In this way, a poor node may be able to crawl all the way to the top. The competition induced by the limited capacity of response of the nodes in the network leads also to the downfall of some (possibly leaders) nodes as illustrated in Fig.~\ref{fig:evenabs}-\ref{fig:evenabs23}.
\begin{figure}
\centering
\begin{subfigure}{.5\textwidth}
  \centering
  \includegraphics[width=1\linewidth]{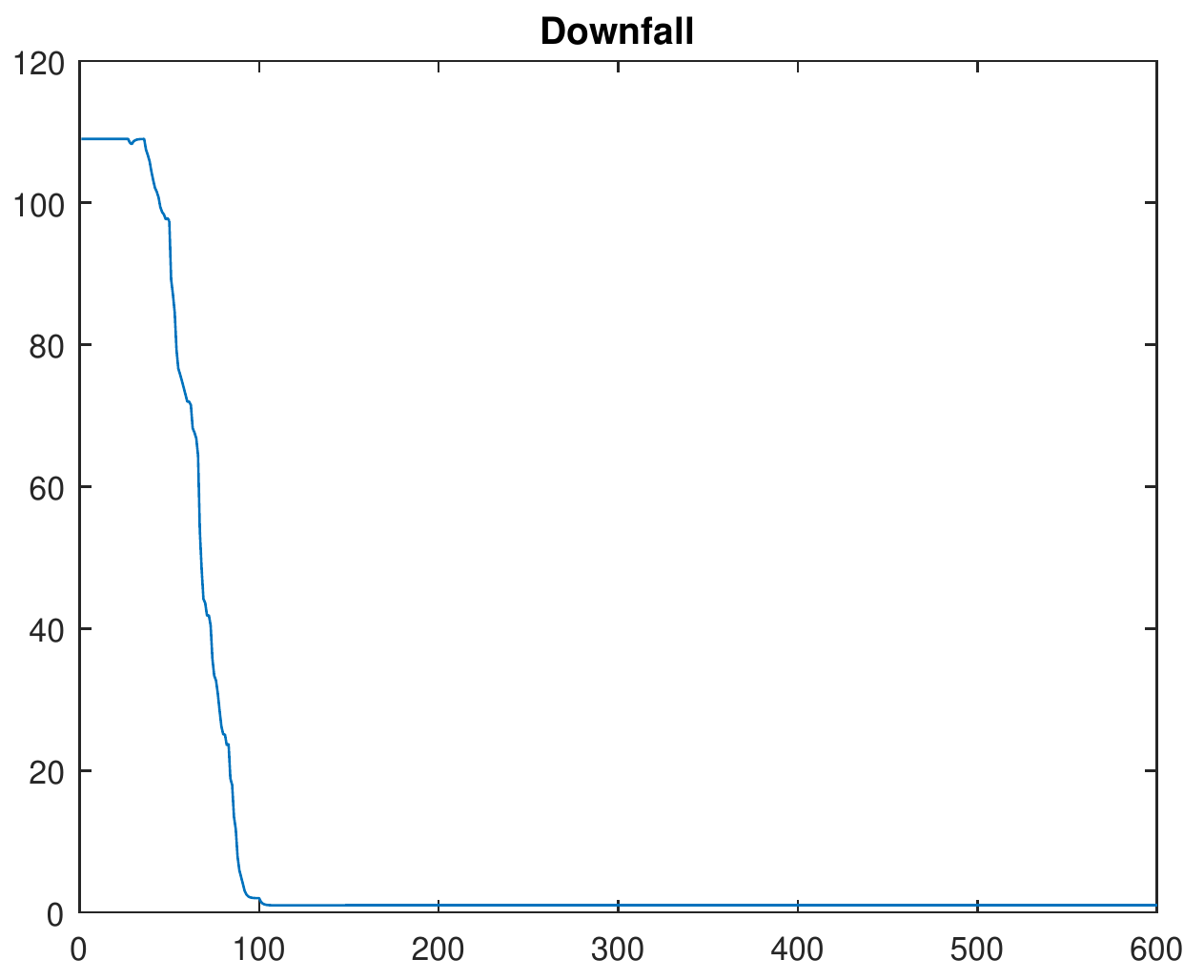}
  \caption{Downfall of a leader.}
  \label{fig:evenabs}
\end{subfigure}%
\begin{subfigure}{.5\textwidth}
  \centering
  \includegraphics[width=1\linewidth]{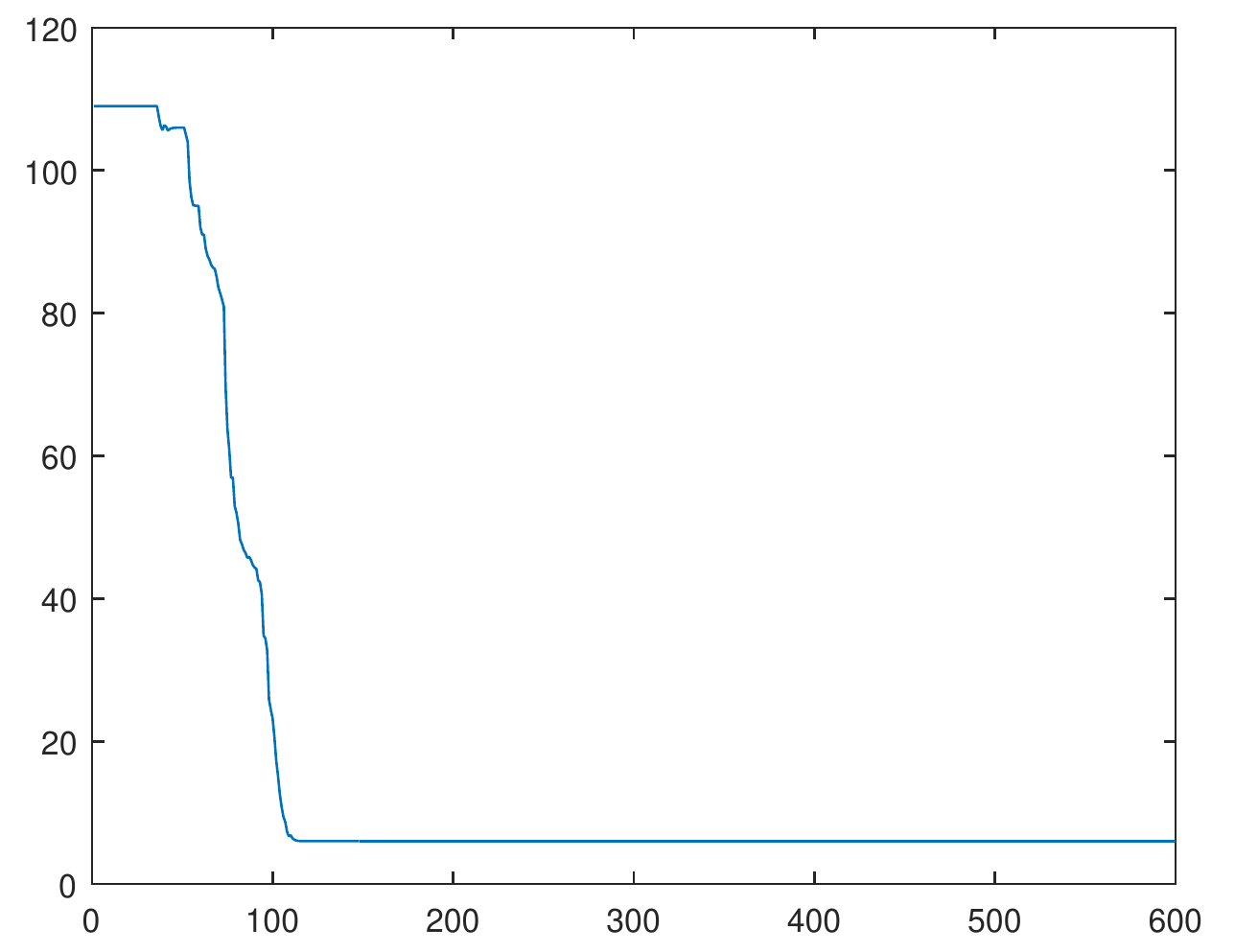}
  \caption{Downfall of a leader}
  \label{fig:evenabs23}
\end{subfigure}
\caption{Biased system with~$\theta=2$. We observe the metastable behavior of a typical downfall of a leader.}
\label{fig:testes}
\end{figure}
This conforms to a metastable behavior as leaders may remain at the top for quite some time until the possible rise of other explorer nodes leads to their descent.

To summarize, in this section, we illustrated that the attractors of the dynamical system may exhibit different macroscopic features based upon the initial distribution of~$\left(\mathbf{p}(n)\right)$, the response distribution, or other parameters. The emergence of few dominant nodes is possible, depending on the configuration of the parameters. Such parameters control how steep the power-law like histograms may be. As we observed in Subsection~\ref{subsec:rise}, the emergence of new leaders and downfall of old ones is well captured by the model.

\section{Concluding Remarks}\label{sec:concluding}


In this paper, we have proposed a family of stochastic dynamical systems model to study systems of interacting agents whose ties vary over time via reinforcement or fading, as observed in many social networks and EMONs of interest. The stochastic dynamical systems constructed bear a special structure in that they converge almost surely to a proper subset of the set of binary matrices (or in other words, directed graphs). As discussed, standard convergence methods for Markov processes or martingales are not applicable to our family of stochastic dynamical systems given by~\eqref{eq:stoc22} and we have developed new techniques that might be of independent interest. We have observed, for instance, that standard absorbing-like arguments for Markov processes (e.g., irreducibility plus ergodicity) do not apply to our setup as the process~$\left(\mathbf{p}(n)\right)$ never coalesces with extreme points~$\mathcal{E}\subset \left\{0,1\right\}^{N \times N}$ in finite time. We have established our convergence results by showing that: \textbf{i)} any~$E_{\ell}\in \mathcal{E}$ is a local attractor with positive probability \textbf{(}\textbf{Theorem}~\ref{co:bounde}\textbf{)}; \textbf{ii)} any arbitrarily small cover of~$\mathcal{E}$ is recurrent \textbf{(}\textbf{Theorem}~\ref{th:ball}\textbf{)}. These two assertions were then combined to prove strong-convergence via Theorem~\ref{th:markovhomog}.

In Section~\ref{sec:simula}, we have illustrated through numerical simulations a variant of behavior and characterization of the associated attractors of the family~\eqref{eq:stoc22}. In particular, we have observed that the shape of the histograms characterizing the out-flow degree of nodes in the limiting graphs depend on the initial conditions, number of nodes and response distribution. The conditions for the convergence of our family of stochastic dynamical systems are broad enough to be applicable to different phenomena, in particular, as we have illustrated in the Subsection~\ref{subsec:rise}, our model with appropriate choice of parameters can exhibit the emergence and downfall of leaders as often observed in social networks. This latter phenomenon empirically captured the fact that high status solely is not stable, even under biased response (based on the out-flow degree), when other nodes perform exploration. In other words, the preferential attachment effect due to the biased response mode is not enough to withstand the exploration factor in other (possibly \emph{poor}) nodes.

It is well known that, depending on the application of interest, networks with certain topological properties present more robust functional behavior. For instance, collaborative networks with positive~$E-I$ index -- amount of connections~$E$ external to an organization or community versus the internal connections~$I$ -- are more adaptable to cooperation and speed up recovery in face of crisis~\cite{stern}. Collaborative networks with hierarchical structures tend to be efficient under crisis as well~\cite{krackhardt}. As future work, we plan to investigate optimal control methods to induce the emergence of attractors with desirable topological properties, as described above.

\small
\bibliographystyle{IEEEtran}
\bibliography{IEEEabrv,biblio}

\begin{thebibliography}{10}
\providecommand{\url}[1]{#1}
\csname url@samestyle\endcsname
\providecommand{\newblock}{\relax}
\providecommand{\bibinfo}[2]{#2}
\providecommand{\BIBentrySTDinterwordspacing}{\spaceskip=0pt\relax}
\providecommand{\BIBentryALTinterwordstretchfactor}{4}
\providecommand{\BIBentryALTinterwordspacing}{\spaceskip=\fontdimen2\font plus
\BIBentryALTinterwordstretchfactor\fontdimen3\font minus
  \fontdimen4\font\relax}
\providecommand{\BIBforeignlanguage}[2]{{%
\expandafter\ifx\csname l@#1\endcsname\relax
\typeout{** WARNING: IEEEtran.bst: No hyphenation pattern has been}%
\typeout{** loaded for the language `#1'. Using the pattern for}%
\typeout{** the default language instead.}%
\else
\language=\csname l@#1\endcsname
\fi
#2}}
\providecommand{\BIBdecl}{\relax}
\BIBdecl

\bibitem{viralyoutube}
T.~Broxton, Y.~Interian, J.~Vaver, and M.~Wattenhofer, ``Catching a viral
  video,'' \emph{Journal of Intelligent Information Systems}, pp. 1--19, 2011.

\bibitem{krishnan}
\BIBentryALTinterwordspacing
L.~Ma, R.~Krishnan, and A.~L. Montgomery, ``Latent homophily or social
  influence? {A}n empirical analysis of purchase within a social network,''
  \emph{Management Science}, vol.~61, no.~2, pp. 454--473, Feb. 2015. [Online].
  Available: \url{http://dx.doi.org/10.1287/mnsc.2014.1928}
\BIBentrySTDinterwordspacing

\bibitem{yingdalu}
\BIBentryALTinterwordspacing
Y.~Lu, K.~Jerath, and P.~V. Singh, ``The emergence of opinion leaders in a
  networked online community: A dyadic model with time dynamics and a heuristic
  for fast estimation,'' \emph{Management Science}, vol.~59, no.~8, pp.
  1783--1799, 2013. [Online]. Available:
  \url{http://dx.doi.org/10.1287/mnsc.1120.1685}
\BIBentrySTDinterwordspacing

\bibitem{ritahurricane}
\BIBentryALTinterwordspacing
N.~Kapucu and V.~Garayev, ``Collaborative decision-making in emergency and
  disaster management,'' \emph{International Journal of Public Administration},
  vol.~34, no.~6, pp. 366--375, 2011. [Online]. Available:
  \url{http://dx.doi.org/10.1080/01900692.2011.561477}
\BIBentrySTDinterwordspacing

\bibitem{anisya}
A.~Thomas and L.~Fritz, ``Disaster {R}elief, inc.'' \emph{Harvard {B}usiness
  {R}eview}, vol.~84, no.~11, pp. 114--122, 2006.

\bibitem{katrina}
C.~T. Butts, R.~M. Acton, and C.~S. Marcum, ``Interorganizational collaboration
  in the hurricane {K}atrina response,'' \emph{International Network for Social
  Network Analysis}, vol.~13, no.~1, pp. 1--36, December 2011.

\bibitem{adaptivenet}
\BIBentryALTinterwordspacing
T.~Gross and B.~Blasius, ``Adaptive coevolutionary networks: a review,''
  \emph{Journal of The Royal Society Interface}, vol.~5, no.~20, pp. 259--271,
  2008. [Online]. Available:
  \url{http://rsif.royalsocietypublishing.org/content/5/20/259}
\BIBentrySTDinterwordspacing

\bibitem{adaptive2}
T.~Gross and H.~Sayama, \emph{Adaptive {N}etworks}.\hskip 1em plus 0.5em minus
  0.4em\relax Springer-Verlag Berlin Heidelberg, 2009.

\bibitem{adaptive3}
H.~Sayama, I.~Pestov, J.~Schmidt, B.~J. Bush, C.~Wong, J.~Yamanoi, and
  T.~Gross, ``Modeling complex systems with adaptive networks,''
  \emph{Computers and Mathematics with Applications}, vol.~65, pp. 1645--1664,
  2013.

\bibitem{adaptive4}
\BIBentryALTinterwordspacing
T.~Gross, C.~J.~D. D'Lima, and B.~Blasius, ``Epidemic dynamics on an adaptive
  network,'' \emph{Phys. Rev. Lett.}, vol.~96, p. 208701, May 2006. [Online].
  Available: \url{http://link.aps.org/doi/10.1103/PhysRevLett.96.208701}
\BIBentrySTDinterwordspacing

\bibitem{kolar}
\BIBentryALTinterwordspacing
M.~Kolar, L.~Song, A.~Ahmed, and E.~P. Xing, ``Estimating time-varying
  networks,'' \emph{Ann. Appl. Stat.}, vol.~4, no.~1, pp. 94--123, 03 2010.
  [Online]. Available: \url{http://dx.doi.org/10.1214/09-AOAS308}
\BIBentrySTDinterwordspacing

\bibitem{Xing}
\BIBentryALTinterwordspacing
S.~Hanneke, W.~Fu, and E.~P. Xing, ``Discrete temporal models of social
  networks,'' \emph{Electron. J. Stat.}, vol.~4, pp. 585--605, 2010. [Online].
  Available: \url{http://dx.doi.org/10.1214/09-EJS548}
\BIBentrySTDinterwordspacing

\bibitem{Crane}
\BIBentryALTinterwordspacing
H.~Crane, ``Time-varying network models,'' \emph{Bernoulli}, vol.~21, no.~3,
  pp. 1670--1696, 2015. [Online]. Available:
  \url{http://dx.doi.org/10.3150/14-BEJ617}
\BIBentrySTDinterwordspacing

\bibitem{Crane2}
\BIBentryALTinterwordspacing
------, ``Dynamic random networks and their graph limits,'' \emph{Ann. Appl.
  Probab.}, vol.~26, no.~2, pp. 691--721, 04 2016. [Online]. Available:
  \url{http://dx.doi.org/10.1214/15-AAP1098}
\BIBentrySTDinterwordspacing

\bibitem{hebb}
\BIBentryALTinterwordspacing
G.~Shaw, ``\BIBforeignlanguage{English}{Donald {H}ebb: The organization of
  behavior},'' in \emph{\BIBforeignlanguage{English}{Brain Theory}}, G.~Palm
  and A.~Aertsen, Eds.\hskip 1em plus 0.5em minus 0.4em\relax Springer Berlin
  Heidelberg, 1986, pp. 231--233. [Online]. Available:
  \url{http://dx.doi.org/10.1007/978-3-642-70911-1_15}
\BIBentrySTDinterwordspacing

\bibitem{debora1}
\BIBentryALTinterwordspacing
D.~M. Gordon, \emph{American Scientist}, no.~1. [Online]. Available:
  \url{http://www.jstor.org/stable/29775362}
\BIBentrySTDinterwordspacing

\bibitem{antevolution}
A.~F.~G. Bourke and N.~R. {F}ranks, \emph{Social Evolution in Ants.}\hskip 1em
  plus 0.5em minus 0.4em\relax {P}rinceton {U}niversity Press, 1995.

\bibitem{Dorigo}
\BIBentryALTinterwordspacing
M.~Dorigo and C.~Blum, ``Ant colony optimization theory: A survey,''
  \emph{Theor. Comput. Sci.}, vol. 344, no. 2-3, pp. 243--278, Nov. 2005.
  [Online]. Available: \url{http://dx.doi.org/10.1016/j.tcs.2005.05.020}
\BIBentrySTDinterwordspacing

\bibitem{williamsprobability}
\BIBentryALTinterwordspacing
D.~Williams, \emph{Probability with Martingales}, ser. Cambridge {M}athematical
  {T}extbooks.\hskip 1em plus 0.5em minus 0.4em\relax Cambridge University
  Press, 1991. [Online]. Available:
  \url{http://books.google.pt/books?id=RnOJeRpk0SEC}
\BIBentrySTDinterwordspacing

\bibitem{Diffusion}
L.~C.~G. Rogers and D.~Williams, \emph{Diffusions, Markov Processes and
  Martingales: Foundations}, 2nd~ed., ser. Cambridge Mathematical
  Library.\hskip 1em plus 0.5em minus 0.4em\relax Cambridge, UK: Cambridge
  University Press, April 2000, vol.~1.

\bibitem{billi}
P.~Billingsley, \emph{Convergence of Probability Measures}, 2nd~ed., ser.
  Probability and Statistics.\hskip 1em plus 0.5em minus 0.4em\relax Wiley,
  August 1999.

\bibitem{convex}
J.-B. Hiriart-Urruty and C.~Lemar\'{e}chal, \emph{Fundamentals of Convex
  Analysis}, ser. Grundlehren Text Editions.\hskip 1em plus 0.5em minus
  0.4em\relax Springer-Verlag Berlin Heidelberg, 2001.

\bibitem{stern}
D.~Krackhardt and R.~N. Stern, ``Informal networks and organizational crises:
  An experimental simulation,'' \emph{Social {P}sychology {Q}uarterly}, pp.
  123--140, 1988.

\bibitem{krackhardt}
D.~Krackhardt and J.~R. Hanson, ``Informal networks,'' \emph{Harvard {B}usiness
  {R}eview}, vol.~71, no.~4, pp. 104--111, 1993.

\end{thebibliography}

\end{document}